\documentclass[12pt]{article}
\pdfoutput=1

\newlength{\abstractwidth}
\usepackage{amscd,amsmath,amssymb,amsfonts,xspace,mathrsfs,amsthm,dsfont,bbm}
\usepackage{color}

\usepackage{bbold}
\usepackage{latexsym}
\usepackage{graphicx}
\usepackage{dsfont}
\usepackage{color}
\usepackage{longtable}
\usepackage{hyperref}

\newtheorem{proposition}{Proposition}

\newtheorem{conj}{Conjecture}
\newtheorem{prob}{Problem}

\numberwithin{equation}{section}

\def\bea{\begin{eqnarray}}
\def\eea{\end{eqnarray}}
\def\be{\begin{equation}}
\def\ee{\end{equation}}
\def\ba{\begin{align}}
\def\ea{\end{align}}
\def\bse{\begin{subequations}}
\def\ese{\end{subequations}}

\newcommand{\nn}{\nonumber}

\newcommand{\Li}{{\rm Li}}

\def\Im{\,{\rm Im}\,}
\def\Re{\,{\rm Re}\,}

\def\({\left(}
\def\){\right)}
\def\[{\left[}
\def\]{\right]}
\def\<{\left\langle}
\def\>{\right\rangle}
\def\hf{{1\over 2}}

\newcommand{\p}{\partial}

\newcommand{\de}{\mathrm{d}}

\newcommand{\I}{\mathrm{i}}

\def\vph{\varphi}

\newcommand{\eps}{\epsilon}
\newcommand{\veps}{\varepsilon}
\newcommand{\vth}{\vartheta}

\newcommand{\cL}{\mathcal{L}}

\newcommand{\cF}{\mathcal{F}}

\newcommand{\cS}{\mathcal{S}}

\newcommand{\cM}{\mathcal{M}}

\newcommand{\cN}{\mathcal{N}}

\newcommand{\cX}{\mathcal{X}}

\newcommand{\cT}{\mathcal{T}}
\newcommand{\cJ}{\mathcal{J}}
\newcommand{\cZ}{\mathcal{Z}}
\newcommand{\cI}{\mathcal{I}}
\newcommand{\cO}{\mathcal{O}}

\newcommand{\cY}{\mathcal{Y}}

\newcommand{\IR}{\mathds{R}}
\newcommand{\IC}{\mathds{C}}
\newcommand{\IZ}{\mathds{Z}}
\newcommand{\IQ}{\mathds{Q}}

\newcommand{\IP}{\mathds{P}}

\newcommand{\sgn}{\mbox{sgn}}

\newcommand{\tn}{\tilde n}

\def\tcJ{\tilde\cJ}

\def\cl0{\tilde c_0}

\newcommand{\thb}{\vth_\beta}
\newcommand{\vbt}{v_\beta}

\def\bOm{\bar\Omega}

\def\ba{\bar a}

\def\hv{\hat v}
\def\hw{\hat w}

\def\hvb{\hat v_\beta}

\def\hsig{\hat\sigma}

\def\cXsf{\cX^{\rm sf}}

\def\Br{^{\rm Br}}

\def\fl#1{\lfloor #1\rfloor}

\def\Pv{\mbox{P.v.}}

\def\cYa{\cY^{(1)}_\beta}
\def\cYb{\cY^{(2)}_\beta}
\def\cYc{\cY^{(2)}_0}
\def\cYai#1{\cY^{(1,#1)}_\beta}
\def\cYbi#1{\cY^{(2,#1)}_\beta}
\def\cYci#1{\cY^{(2,#1)}_0}

\def\under#1#2{\mathop{#1}\limits_{#2}}

\begin{document}

\thispagestyle{empty}

\begin{flushright}
arXiv:2106.12006v3
\end{flushright}
\vskip 0.3in

\begin{center}
{\Large \bf Conformal TBA for resolved conifolds}
\vskip 0.2in

{\large Sergei Alexandrov$^\dagger$ and Boris Pioline$^\ddagger$}

\vskip 0.3in

$^\dagger$ {\it
Laboratoire Charles Coulomb (L2C), Universit\'e de Montpellier,
CNRS, \\ F-34095, Montpellier, France}

\vskip 0.15in

$^\ddagger$ {\it Laboratoire de Physique Th\'eorique et Hautes
Energies (LPTHE), UMR 7589 CNRS-Sorbonne Universit\'e,
Campus Pierre et Marie Curie,
4 place Jussieu, \\ F-75005 Paris, France} \\

\vskip 0.15in

{\tt \small sergey.alexandrov@umontpellier.fr,  pioline@lpthe.jussieu.fr}

\vskip 0.4in

\begin{abstract}
\vskip 0.1in

We revisit the Riemann-Hilbert problem determined by Donaldson-Thomas invariants for
the resolved conifold and for other small crepant resolutions.
While this problem can be recast  as a system of TBA-type equations in
the conformal limit, solutions are ill-defined due to divergences in the sum over
infinite trajectories in the spectrum of D2-D0-brane bound states.
We explore various prescriptions to make the sum well-defined, show that one of them reproduces the existing solution
in the literature, and identify an alternative solution which is better behaved in a certain limit.
Furthermore, we show that a suitable asymptotic expansion of the $\tau$ function
reproduces the genus expansion of the topological string partition function
for any small crepant resolution.
As a by-product, we conjecture new
integral representations for the triple sine function, similar to Woronowicz' integral representation
for Faddeev's quantum dilogarithm.
\end{abstract}
\end{center}

\newpage
\setcounter{tocdepth}{2}
\tableofcontents

\baselineskip=15pt
\setcounter{page}{1}
\setcounter{equation}{0}
\setcounter{footnote}{0}
\setlength{\parskip}{0.2cm}

\section{Introduction}

The possibility to construct a smooth hyperk\"ahler or quaternion-K\"ahler space $\cM$ from
a suitable collection of integers $\{\Omega(\gamma,z)\}$ has been appreciated for a long time
in the physics literature on quantum field theories or string vacua with $\cN=2$ supersymmetry in four dimensions:
in this context, the integers $\Omega(\gamma,z)$ count supersymmetric bound states with fixed charge
$\gamma$ and asymptotic value $z$ of the moduli in four dimensions, and the quaternionic
space $\cM$ appears as the moduli space of massless fields after reduction on a circle
\cite{Gaiotto:2008cd,Alexandrov:2008gh,Alexandrov:2011va}.
Recently, Tom Bridgeland
has proposed a variant of this idea in the mathematically rigorous framework of triangulated categories
of Calabi-Yau-3 type:  the integers $\Omega(\gamma,z)$ are the generalized Donaldson-Thomas (DT) invariants
for the stability condition coordinatized by $z$ in the stability space $\cS$, while
$\cM$ is the total space of the holomorphic tangent bundle $T\cS$, endowed with
a complex hyperk\"ahler metric. The latter is  determined by solutions to a suitable Riemann-Hilbert (RH) problem
\cite{Bridgeland:2019fbi,Bridgeland:2020zjh}, which arises formally as the conformal limit of the
RH problem considered by physicists \cite{Gaiotto:2014bza}.
The smoothness of the metric across walls of marginal stability
in $\cS$ is ensured by the usual wall-crossing formula for DT invariants \cite{ks}.

In a recent work \cite{Alexandrov:2021wxu}, using twistorial techniques developed in the physics literature
\cite{Alexandrov:2008ds,Alexandrov:2008nk,Alexandrov:2011va},
we reformulated Bridgeland's RH problem as a system of TBA-type integral equations
for Darboux coordinates $\cX_\gamma(z,\theta,t)$ on $\cM$, where $\theta$ parametrizes
the fiber of the tangent bundle and $t\in \IC^\times$ is the twistorial parameter. Using this reformulation,
we obtained general formulae  for the Plebanski potential
$W(z,\theta)$ and the Joyce potential
$F(z,\theta)$ in terms of solutions to that system. These functions, which are homogeneous
under the $\IC^\times$ action on $\cS$, are required
to satisfy `heavenly' and `isomonodromy' equations, respectively,
and determine a complex hyperk\"ahler cone metric on $\cM$. For general coupled BPS structures where the
antisymmetric pairing $\langle \gamma,\gamma'\rangle$ is non-vanishing for some pair of
active charge vectors (i.e. $\Omega(\gamma) \Omega(\gamma') \neq 0$), the TBA equations
are highly non-trivial, but they
can be solved iteratively providing formal power series representations for $\cX_\gamma$,
$W$ and $F$.
In the uncoupled, finite case (where only a finite number of DT invariants $\Omega(\gamma)$
are non-zero, and moreover $\langle \gamma,\gamma'\rangle=0$ for any pair of them),
the TBA equations trivialize and lead to elementary formulae for $W$ and $F$, and somewhat less
elementary but explicit formulae for the Darboux coordinates $\cX_\gamma$ as linear combinations
of the generalized Gamma function $\Lambda(z,\eta)$.
In addition, the latter can all be obtained as suitable $z$-derivatives of a function $\tau(z,\theta,t)$,
which is given by a linear combination of generalized Barnes functions $\Upsilon(z,\eta)$
(see Appendix \ref{ap-functions} for the definition of these special functions).

In this work, we address an intermediate case, namely that of uncoupled but infinite BPS structures.
This case is relevant for the derived category of coherent sheaves on a non-compact Calabi-Yau threefold $X$
with no compact divisors, such as the conifold or more generally, small crepant resolutions of CY3 singularities.
For such cases, the pairing $\langle \gamma,\gamma'\rangle$ vanishes for any pair of active charge vectors,
and moreover $\Omega(\gamma)$ is non-zero (and independent of $k$) on a finite set of `trajectories' of the form
$\gamma=\beta+k\delta$ where $\delta$ is the charge vector of a point (corresponding to the skyscraper sheaf on $X$),
$k\in\IZ$ and $\beta$ runs over the set $B=\{\pm \beta_1,\dots,\pm \beta_n\}$ where $\beta_\ell$
are the homology classes of rational curves in $X$.
The integers $n_\beta= \Omega(\beta+k\delta)$ coincide with the genus-zero Gopakumar-Vafa (GV) invariants,
and are known explicitly for all small crepant resolutions of toric CY3 singularities
(see e.g. \cite[\S 5]{Mozgovoy:2020has} and references therein).
Moreover, $\Omega(k\delta)$ is also independent of $k$ and equal to minus the Euler number $\chi_X$ of $X$.
In physical terms, $\Omega(\beta+k\delta)$ counts BPS bound states of a D2-brane
wrapped on a rational curve with $k$ units of D0-brane charge.
Importantly, in such uncoupled cases the Darboux coordinates can still
be derived from a $\tau$ function, even though the BPS spectrum is infinite.

The special case of the resolved conifold $X=\cO(-1,-1)\rightarrow \IP^1$ was considered in
\cite{Bridgeland:2017vbr,Bridgeland:2019fbi}, where solutions to the Riemann-Hilbert problem
were found, for vanishing and generic values of the fiber coordinates $\theta$, respectively.
Using these explicit solutions, the Plebanski potential (called there `Joyce function')
was obtained  in \cite{Bridgeland:2019fbi}, while an expression for the $\tau$ function was derived
in \cite{Bridgeland:2017vbr}  (only for $\theta=0$).
Strikingly, the asymptotic expansion of $\tau(z,0,t)$ at small $t$ turns out to match
the genus expansion of the topological string partition function on $X$, opening the
possibility that the $\tau$ function may provide a natural non-perturbative definition of
the topological string partition function (see also
\cite{Koshkin:2007mz,Pasquetti:2009jg,Grassi:2014zfa,
Hatsuda:2015owa,Krefl:2015vna,Bonelli:2016qwg,
Coman:2018uwk,Grassi:2019coc,Alim:2021lld}
for other viewpoints on this problem).

\subsubsection*{Summary of the main results}

In this work, we revisit Bridgeland's results
for the conifold by using the TBA formulation of the RH problem,
and extend them to arbitrary small crepant resolutions. As for any uncoupled BPS structure,
the TBA equations trivialize and produce explicit expressions for the Darboux coordinates and
the potentials.
However, in contrast to the finite case, these expressions are no longer well-defined
because the sums over charges $\gamma$ are {\it not} absolutely convergent.
Thus, the problem reduces to finding a prescription for these divergent sums which
leads to sensible results
satisfying the properties required from solutions of the RH problem.

First, we attempt to define the Darboux coordinates and the $\tau$ function by summing the
finite case result  over the infinite set of non-vanishing DT invariants $\Omega(\beta+k\delta)$.
For the conifold and for vanishing fiber coordinates $\theta$, this approach works and
reproduces the solution found in \cite{Bridgeland:2017vbr},
upon using a Binet-type integral representation for the special functions $\Lambda(z,\eta)$ and $\Upsilon(z,\eta)$.
However, it breaks down when the fiber coordinates are non-zero, or for general small crepant resolutions not satisfying the
property $\sum_{\beta\in B} n_\beta=\chi_X$, which holds
coincidentally for the  conifold. Indeed, in those cases the prescription leads to integrals which diverge at the origin.
Most of these integrals can be made convergent by shifting by hand the contour so as
to avoid the origin.
Following this {\it ad hoc} procedure and dropping remaining divergent contributions, it is possible to arrive at
the solution for the conifold with generic fiber coordinates given in \cite{Bridgeland:2019fbi},
and generalize it to other small crepant resolutions.
However, such procedure is hard to justify and cannot be accepted as a consistent prescription.
Furthermore, it is impossible to make sense of the Plebanski potential $W$ in this way.

Next, we explore an alternative approach. Rather than taking the results for finite BPS structures as a starting point,
we take a few steps back. First, we do {\it not} combine the contributions of opposite charges $\gamma$ and $-\gamma$,
which was an important step in obtaining the finite case solution, and second, we express all relevant quantities
in terms of rational DT invariants
\be
\bOm(\gamma)=\sum_{d|\gamma} \frac{1}{d^2}\, \Omega(\gamma/d).
\label{defbOm}
\ee
In the case of interest, they are supported on charges $\gamma=n(\beta+k\delta)$ for any $n\geq 1$.
Then, we first make sense of the sum over $k$ for each trajectory labelled by $\beta$ (including the one with $\beta=0$)
and only at the last step perform the sum over $n$. It turns out that the first step boils down to
making sense of sums of the form
\be
S(v,w,\vth,\vph)=\sum_{k\in\IZ}^\infty \frac{ e^{2\pi\I (\vth+k \vph)}}{v+k w}.
\label{seriesk}
\ee
This sum is not absolutely convergent, even when $\vph=0$. However, in this case
it becomes absolutely convergent after combining contributions of positive and negative $k$
and leads to a solution that coincides with the one in \cite{Bridgeland:2017vbr} in the conifold case.
For non-vanishing $\vph$, we shall argue that the sum \eqref{seriesk} can be recast into the following integral
\be
S(v,w,\vth,\vph)=\frac{e^{2\pi\I (\vth- \vph v/w)}}{w}\int_0^\infty \frac{\de y}{y} \, \frac{y^{v/w}}{1-y}\, ,
\label{series4}
\ee
where the divergence at $y=1$ is avoided either by moving the contour above or
below the pole, or by using the principal value prescription.
This gives three possible results $S^{(\eps)}$ labelled by $\eps=\pm$ and 0, respectively.
All of them lead to well-defined Darboux coordinates, $\tau$ function and Plebanski potential.

For instance, let us denote $Z_\beta=\vbt$, $Z_\delta=w$ where $Z_\gamma$ is the central charge function,
while $\thb$ and $\vph$ be the corresponding fiber coordinates (see \$\ref{subsec-CYs} for precise definitions).
Then, the principle value prescription results in the following $\tau$ function
\bea
\log\tau&=&
\sum_{\beta\in B}n_\beta\[\frac{t}{4\pi w}\(\cJ_{2,0}\(\frac{w}{t},\frac{\vbt}{t},\vph,\thb\)
+2\pi \cJ_{1,1}\(\frac{w}{t},\frac{\vbt}{t},\vph,\thb\)\)-\frac{1}{4\pi^2}\, \cI_2\(\frac{\vbt}{w},\vph,\thb\) \]
\nn\\
&&
-\frac{\chi_X t}{4\pi w} \[\tcJ_{2,0}\(\frac{w}{t},\vph\)
+2\pi \tcJ_{1,1}\(\frac{w}{t},\vph\)\]-\frac{\chi_X}{24}\, \log \frac{w}{t}\, ,
\label{reztau-pv}
\eea
where the special functions $\cJ_{m,n}$, $\tcJ_{m,n}$ and $\cI_n$ are introduced
in \eqref{deffunJ}, \eqref{deffunJ0} and \eqref{def-cIn}, respectively.
The Plebanski potential is in turn given
(up to irrelevant terms linear in the fiber variables) in terms of Bernoulli polynomials,
\be
W=\frac{\pi^2}{6\I w} \sum_{\beta\in B}n_\beta\, \cot(\pi \vbt/w) \, B_3\(\[\thb-\vph\, \frac{\vbt}{w}\]\),
\label{resW-pv}
\ee
where the bracket notation $[x]$ is defined in \eqref{defbr}.
The solutions for different prescriptions differ from the functions given above by simple contributions
expressed in terms of polylogarithms.
For example, the difference between the Plebanski potential corresponding to the prescription leading
to  Bridgeland's solution in the case of the conifold \cite{Bridgeland:2019fbi}
(see \eqref{WBr}) and the one given in \eqref{resW-pv},
is equal to
\be
W\Br-W=\frac{ \pi^2 }{3 w}\sum_{\beta\in B^+}n_\beta \,B_3\(\[\thb-\vph\, \frac{\vbt}{w}\]\),
\ee
where two three-logarithms have been combined using the identity \eqref{Li-ident}.

Thus, all ambiguities in the definition of
the RH problem appear to be captured by the choice of prescription for the integral \eqref{series4}.
In principle, this choice can be done separately for each $\beta\in H_2\cup\{0\}$.
In practice, a random assignment of $\eps_\beta$ is of course unreasonable.
However, it is important to take this possibility into account because, as we will argue,
the choice $\eps_\beta=\sgn\Im (\vbt/w)$ (i.e. corresponding to distinct prescriptions for three classes of trajectories
$\{\beta=0,B^+,B^-\}$ where $B^\pm$ are defined in \eqref{defH2p})
precisely leads to Bridgeland's solution for the conifold in \cite{Bridgeland:2019fbi}.
Note that for this solution, the Darboux coordinates $\cX_\gamma$ are {\it not} meromorphic in $t$
(contrary to the claim in \cite{Bridgeland:2019fbi}),
but have an essential singularity at $t=\vph/w$. This in fact holds for any choice of $\epsilon_\beta$.
However, rather than choosing this $\beta$-dependent prescription, we argue that a preferable choice is to take
the principle value prescription for the integral \eqref{series4}. Indeed, this appears to be the only
prescription which leads to a solution which in the limit $\thb=\vph=0$ agrees
with the results in \cite{Bridgeland:2017vbr}, which are fixed uniquely by requiring polynomial behavior at $t\to \infty$.

The agreement between the solution resulting from the choice $\eps_\beta=\sgn\Im (\vbt/w)$
with Bridgeland's solution relies on some novel integral representations
\eqref{WoroF}-\eqref{WoroG2} for the double and triple sine functions
appearing in Bridgeland's solution (and in earlier works on the topological
string partition function on the conifold \cite{Koshkin:2007mz,Krefl:2015vna}).
The double and triple sine functions are usually defined through Barnes' multiple gamma functions
\cite{Barnes1964,Ruijsenaars2000OnBM,Kurokawa2003MultipleSF}, and are closely related
to the quantum dilogarithm and trilogarithm.
The required integral representation for the double sine function is in fact a consequence of Woronowicz'
representation \eqref{WoroPhi} for Faddeev's quantum dilogarithm, as we show in Appendix \ref{ap-Br}.
We state the required integral representations for the triple sine function with two coinciding periods
as Conjecture \ref{conj-sine} in that Appendix. We have checked it numerically for some random values of the arguments,
and expect that it can be generalized to arbitrary periods and  proven along the same lines as \eqref{WoroF}.

\medskip

Finally, we extend the observation of \cite{Bridgeland:2017vbr} about a relation between a perturbative expansion
of the $\tau$ function and the topological string partition function.
Namely, we show that for all small crepant resolutions the expansion of \eqref{reztau-pv} at small $t$ and $\thb=\vph=0$
reproduces al terms with $g\ge 1$ in the genus expansion of the topological string,
while the genus 0 contribution can be extracted from the Plebanski potential \eqref{resW-pv}.
Moreover, we use this result as an additional argument in favor of the solution produced by the principle value prescription.

\medskip

The organization of the paper is as follows.
In the next section we review the RH problem for general BPS structures and its solution in the uncoupled BPS case,
and specify the RH
problem for the class of non-compact CY threefolds given by small crepant resolutions.
In \S\ref{sec_pass1} we present the results of our first approach where we try to sum up the known solution for
a finite uncoupled BPS structure. In \S\ref{sec_pass2} we derive new solutions following an alternative approach based
on summing up  contributions of rational DT invariants and discuss their properties.
Finally, in \S\ref{sec_top} we discuss the relation between the $\tau$ function, Plebanski potential
and the topological string partition function. Appendix \ref{ap-functions} collects definitions of various special functions
as well as their properties and various integral identities. Appendices \ref{ap-Br} and \ref{ap-new} present
special functions appearing in the solutions found in \S\ref{sec_pass1} and \S\ref{sec_pass2}, respectively.
Finally, Appendix \ref{ap-conifold} provides details
of calculations leading to the results presented in \S\ref{sec_pass1}.

\section*{Acknowledgements}

We are grateful to Tom Bridgeland and Joerg Teschner
for very useful correspondence. As this work was being completed,
we learnt that M. Alim, A. Saha and I. Tulli were investigating a closely related problem \cite{Alim:2021gtw}.
We thank them for agreeing to coordinate the release on arXiv.

\section{RH problem for small crepant resolutions}
\label{sec_start}

In this section, we recall the basic Riemann-Hilbert problem associated to a BPS structure,
its reformulation as a set of TBA-like integral equations, and its solution for an
uncoupled BPS structure, following \cite{Bridgeland:2016nqw} and our previous work \cite{Alexandrov:2021wxu}.
In \S\ref{subsec-CYs} we specialize to the case of small crepant
resolutions, where the BPS structure is uncoupled but not finite.

\subsection{RH problem in the uncoupled case}
\label{sec_RH}

For our purposes, a BPS structure consists of a local system of  rank $2N$  lattices $\Gamma_Z$,
equipped with a non-degenerate skew-symmetric pairing $\langle \cdot,\cdot\rangle$,
fibered over the space of stability conditions $\cS$, and equipped with a map $\Omega_Z:\Gamma_\star\to \IZ$
with $\Gamma_\star=\Gamma\backslash\{0\}$
satisfying various axioms \cite{Bridgeland:2016nqw}.  In particular, $\Omega_Z(\gamma)$
should be invariant under $\gamma\mapsto -\gamma$, and satisfy the Kontsevich-Soibelman wall-crossing formula
across walls of marginal stability in $\cS$.
The subscript $Z$ refers to the central charge function (a linear map
$Z:\Gamma\ni\gamma\mapsto Z_\gamma\in\IC$), which is part of the stability data
at any point on $\cS$. In the following we will drop the subscript $Z$, and use complex
coordinates $z^a=Z_{\gamma^a}$ on $\cS$ where $\{\gamma^a\}_{a=1}^{2N}$ is a basis of the lattice
$\Gamma$.

Given a BPS structure, one defines the following Riemann-Hilbert problem:
finding functions $\cX_\gamma(z,\theta,t)$ (called holomorphic Fourier modes,
or Darboux coordinates for brevity)
on $ T\cS\times \IP^1$, where $z$ is a point on the space of stability conditions $\cS$,
$\theta$ parametrizes the fiber of the
holomorphic tangent bundle and $t$ is a coordinate on $\IP^1$, such that
\be
\label{qtorus}
\cX_\gamma\, \cX_{\gamma'} = (-1)^{\langle\gamma,\gamma'\rangle} \cX_{\gamma+\gamma'}
\ee
for all $\gamma,\gamma'\in\Gamma_\star$, $(z,\theta)\in T\cS$ and $t\in\IP^1$
away from the BPS rays $\ell_{\gamma'}$ defined by
\be
\ell_{\gamma'}=\{t\in \IP^1: \ Z_{\gamma'}/t\in \I\IR^-\}.
\label{BPSray}
\ee
Across such a ray, we require that the Fourier modes $\cX_\gamma^\pm$ defined on the two sides of the BPS ray
(on clockwise and anticlockwise sides, respectively) are related by a KS symplectomorphism \cite{ks}
\be
\label{Xgdisc}
\cX_\gamma^+ = \cX_{\gamma}^- (1-\cX_{\gamma'}^-)^{\Omega(\gamma') \langle\gamma,\gamma'\rangle}.
\ee
Moreover, for $t\to 0$ we require that $\cX_\gamma$ reduces, up to exponential corrections,
to
\be
\cXsf_\gamma=\sigma_\gamma\, e^{2\pi\I(\theta_{\gamma} - Z_\gamma/t)},
\qquad
\theta_\gamma=\langle \gamma,\theta\rangle,
\label{cXsf}
\ee
where $\sigma_\gamma$ is a quadratic refinement of the skew-symmetric pairing.\footnote{
Namely a map $\sigma: \Gamma\to \IC^\times$   such that
$\sigma_\gamma\, \sigma_{\gamma'} = (-1)^{\langle\gamma,\gamma'\rangle} \sigma_{\gamma+\gamma'}$. }
For large $|t|$, we require that $\cX_\gamma(t)$ behaves polynomially, i.e.
\be
|t|^{-k}<|\cX_\gamma(t)|<|t|^k
\ee
for some $k>0$.
Finally, these conditions are supplemented by the so called reflection property which
originates in the symmetry of the DT invariants under $\gamma\to-\gamma$
and requires that \cite{Bridgeland:2016nqw}
\be
\cX_{-\gamma}(z,\hsig(\theta),-t)=\cX_\gamma(z,\theta,t),
\label{reflect-cX}
\ee
where $\hsig$ is the twisted inverse map. This map has a simple action
on the shifted  fiber coordinates $ \vth_\gamma$ defined (modulo integers) by absorbing the quadratic refinement,
\be
e^{2\pi\I \vth_\gamma}=\sigma_{\gamma}\, e^{2\pi\I \theta_\gamma}.
\label{defvth}
\ee
namely it acts by $\hsig(\vth_\gamma)=-\vth_\gamma$. Note however that $ \vth_\gamma+\vth_{\gamma'}
\neq \vth_{\gamma+\gamma'}$ modulo integers unless $\langle\gamma,\gamma'\rangle=0$.
In the following we will mostly work in terms of the shifted coordinates $\vth_\gamma$.

Given a solution to this RH problem, one constructs a complex symplectic
form $\omega=\frac{t}{2}\sum_{a,b} \omega_{ab} \de \log\cX_{\gamma_a} \de \log\cX_{\gamma_b}$  on the product
$T\cS\times \IP^1$, where
$ \omega_{ab}$ is the inverse of the integer skew-symmetric matrix
$\omega^{ab} = \langle \gamma^a, \gamma^b\rangle$. The complex HK structure
on $T\cS$ is then read off  from the Laurent expansion
$\omega=t^{-1}\omega_+ -\I \omega_3+t\omega_-$ near $t=0$.
It turns out that the resulting HK metric is encoded in a single function $W(z,\theta)$, called Plebanski potential,
which must satisfy a system of non-linear differential equations,
known as heavenly equations \cite{Bridgeland:2020zjh,Dunajski:2020qhh}.
Furthermore, this function determines a connection on $\IP^1$ such that $\cX_\gamma$ are its flat sections, which
is equivalent to the following horizontal section condition \cite{Bridgeland:2019fbi}
\be
\left[ \frac{\p}{\p z^a}+\frac{1}{t}\, \frac{\p}{\p\theta^a}
-\frac{1}{(2\pi)^2}\sum_{b,c=1}^{2N}\omega^{bc}\frac{\p^2 W}{\p \theta^a\p\theta^b}\,\frac{\p}{\p\theta^c}
\right] \cX_\gamma =0.
\label{actonX1}
\ee
The twistor space $\cZ$ is the quotient of $T\cS\times \IP^1$ by the distribution spanned
by the vector fields \eqref{actonX1} \cite{Dunajski:2020qhh}.

In \cite{Alexandrov:2021wxu}, we have recast the above RH problem
as a TBA-like system of integral equations, arising as a conformal limit \cite{Gaiotto:2014bza}
of the equations encoding instanton corrections in four-dimensional $\cN=2$ supersymmetric
gauge theories compactified on a circle \cite{Gaiotto:2008cd}
or D-instantons in compactifications of string theory on CY threefolds \cite{Alexandrov:2008gh,Alexandrov:2009zh}.
This allowed us to obtain a general formula for the Plebanski potential $W$ in terms of Penrose-like integrals
of the Fourier modes $\cX_\gamma$ \cite[Eq. (1.5)]{Alexandrov:2021wxu}.
In this paper we do not need these general results as we restrict to the case of
uncoupled BPS structures, i.e. such that $\langle \gamma,\gamma'\rangle=0$
whenever $\Omega(\gamma) \Omega(\gamma') \neq 0$.

In the uncoupled case the RH problem drastically simplifies.
Firstly, the DT invariants become independent of the central charge function $Z$.
Secondly, the TBA equations mentioned above become explicit integral {\it expressions}
for the Darboux coordinates
\bea
\cX_\gamma(t)
&=&\cXsf_\gamma(t)\,
\exp\[\frac{1}{2\pi\I}\sum_{\gamma'\in\Gamma_{\!\star}} \bOm(\gamma')
\langle\gamma,\gamma'\rangle\int_{\ell_{\gamma'}}\frac{\de t'}{t'}\,
\frac{t}{t'-t}\,\cXsf_{\gamma'}(t')\],
\label{eqTBA}
\eea
where the r.h.s. is expressed through the rational DT invariants $\bOm(\gamma)$ defined in \eqref{defbOm}.
Thus, there are no equations to solve anymore and all relevant functions can be computed by evaluating Penrose-like integrals.
In particular, for the Plebanski potential one finds the following simple result \cite[Eq. (4.5)]{Alexandrov:2021wxu}
\be
W=\frac{1}{(2\pi\I)^2}\, \sum_{\gamma\in\Gamma_\star} \bOm(\gamma)\,\frac{e^{2\pi\I \vth_{\gamma}}}{Z_\gamma}.
\label{resW}
\ee

Remarkably, for uncoupled BPS structures the Darboux coordinates \eqref{eqTBA} for different basis vectors $\gamma^a$
of the lattice $\Gamma$
(or rather, their derivatives with respect to $t$)
can all be obtained from a single generating function $\tau(z,\theta,t)$. Namely, the following equations, first
proposed in \cite{Bridgeland:2016nqw}, turn out to be integrable for any uncoupled BPS structure
\be
\sum_{b=1}^{2N} \omega^{ab} \p_{z^b} \log\tau =\p_t \log\(\cX_{\gamma^a}/\cXsf_{\gamma^a}\),
\qquad
\forall a=1,\dots, 2N,
\label{eqtau}
\ee
with the following general result \cite[Eq.(3.28)]{Alexandrov:2021wxu}
\be
\log\tau =
\frac{1}{4\pi^2}\sum_{\gamma\in\Gamma_{\!\star}}\bOm(\gamma)\, e^{2\pi\I\vth_\gamma}
\[\int_{\ell_\gamma} \frac{\de t'}{t'} \,\frac{t}{t'-t}
\(1+2\pi\I\, \frac{Z_\gamma}{t'}\)e^{-2\pi\I Z_\gamma/t'}-\log (Z_\gamma/t)\].
\label{restau}
\ee
It is important to note that the $\tau$ function is defined only up to terms independent of $z^a$.
Such terms can be safely added or dropped without affecting the Darboux coordinates $\cX_\gamma$.

If the BPS structure is finite, i.e. if $\Omega(\gamma)$ vanishes except
for a finite set of vectors $\gamma$, the above formulae can be made even more explicit.
The Darboux coordinates are then expressed in terms of
the generalized Gamma function $\Lambda(z,\eta)$ defined in \eqref{defLambda}
\be
\label{Xuncoupled}
\cX_\gamma =  e^{2\pi\I(\vth_{\gamma} - Z_\gamma/t)}
\prod_{\gamma'\in \Gamma \atop \Re(Z_{\gamma'}/t)>0}
\Lambda\( \frac{Z_{\gamma'}}{t},1-[\vth_{\gamma'}] \)^{\Omega(\gamma') \langle \gamma,\gamma'\rangle},
\ee
the $\tau$ function is given in terms of the generalized Barnes function
$\Upsilon(z,\eta)$ defined in \eqref{defUpsilon}
\be
\log\tau = \sum_{\gamma\in \Gamma \atop \Re(Z_{\gamma}/t)>0}
\Omega(\gamma)\,\log \Upsilon\(\frac{Z_\gamma}{t},1-[\vth_\gamma]\) ,
\label{taufun-res2}
\ee
while the Plebanski potential is a simple sum of Bernoulli polynomials
\be
W=-\frac{\pi\I}{3}\sum_{\gamma\in\Gamma_{\!+} } \frac{\Omega(\gamma)}{Z_\gamma}\, B_3([\vth_\gamma]),
\label{Wfinite}
\ee
where $\Gamma_{\! +}$ is the set of vectors $\gamma\in\Gamma_{\!\star} $ such that
$\Im Z_\gamma>0$ or $Z_\gamma\in \IR^-$.

\subsection{BPS structure for small crepant resolutions}
\label{subsec-CYs}

Let us now consider BPS structure associated to a non-compact CY threefold $X$ obtained as a crepant resolution
of an affine singular variety, such that the exceptional locus contains only a finite number $n$ of
rational curves and no compact divisors. Examples include the conifold obtained by resolving the quadratic singularity
$xy=zw$, or the generalized conifold obtained by resolving $xy=z^{N_0} w^{N_1}$ with
$0\leq N_1\leq N_0$, as well as quotients of the form $\IC^2/G\times \IC$ where $G$
is finite subgroup of $SL(2,\IC)$, or of the form $\IC^3/G$ where $G$ is a finite subgroup of $SO(3)$
or $G=\IZ_2\times\IZ_2$ with the two factors acting by $(-1,-1,1)$ and $(1,-1,1)$.

For such geometries, the DT invariants associated to the triangulated category of coherent sheaves
on $X$ take a very simple form. The  charge lattice decomposes as the direct sum of two rank-$N$ sublattices,
$\Gamma=\Gamma_e\oplus\Gamma_m$ where $\Gamma_e=H_2(X,\IZ)\oplus H_0(X,\IZ)$ and $\Gamma_m$
is the dual lattice $\Gamma_e^*$. Since the 4-cycles and 6-cycle corresponding
to $\Gamma_m$ are non-compact, the corresponding central charges are infinite but it turns out that
they decouple from the Riemann-Hilbert problem. Thus,
the space of stability conditions $\cS$ is effectively an open set inside $\IC^{N}$, parametrized by
the central charges $\{v^i=Z_{\alpha^i}\}$ and $w=Z_\delta$, where
$\{\alpha^i\}_{i=1}^{N-1}$ is a basis of $H_2(X,\IZ)$ and $\delta$ is the generator of $H_0(X,\IZ)$,
corresponding to the class of a point on $X$. Similarly, the fiber of the holomorphic tangent bundle
$T\cS$ can be parametrized by $\vth^i:=\vth_{\alpha^i}$ and $\vph:=\vth_\delta$. We also introduce
$B=\{\pm \beta_\ell\}_{\ell=1}^n$,
the set of the classes of rational curves and anti-rational
curves in $X$, which are integer combinations $\beta_\ell=\sum_i\beta_{\ell,i}\alpha^i$ of the basis vectors,
and its two subsets
\be
B^\pm=\{\beta\in B\ : \ \pm\Im(\vbt/w)>0\}.
\label{defH2p}
\ee
For the generalized conifold, $N=N_0+N_1$ and $n=N(N-1)/2$.

Due to the absence of
compact 4 and 6-cycles, the
DT invariants $\Omega(\gamma)$ vanish unless $\gamma\in \Gamma_e$. In particular, since the skew-symmetric pairing is non-vanishing
only between elements of $\Gamma_e$ and the dual lattice $\Gamma_m=\Gamma_e^\star$,
the corresponding BPS structure is uncoupled and the DT invariants are independent of the central
charge function. In fact, $\Omega(\gamma)$ vanishes unless the charge is of the form $\gamma =k\delta$ with a non-zero integer $k$,
or $\gamma=\beta+k\delta$ where $\beta\in B$ and $k\in \IZ$.
In the latter case, depending whether the normal bundle of the rational curve in class $\pm\beta$ is $\cO(-1)\oplus\cO(-1)$
or $\cO(0)\oplus\cO(-2)$, $\Omega(\gamma)$ is equal to $1$ or $-1$   \cite[\S 5]{Mozgovoy:2020has}.
In either case, it coincides with
the genus zero Gopakumar-Vafa invariant $n_{\beta}$. To summarize,
\be
\label{OmCY}
\Omega(\gamma)=\left\{\begin{array}{lll}
-\chi_X, \quad & \gamma=k\delta, \quad & k\in \IZ\setminus \{0\},
\\
n_\beta, \quad & \gamma=\beta+k\delta, \ & k\in \IZ, \ \beta\in B\, ,
\end{array}\right.
\ee
where $\chi_X$ is the Euler characteristic of $X$ (defined using cohomology with compact support).
Sometimes it will be convenient to combine them by setting $n_0=-\chi_X/2$.
For the resolved  conifold $X=\cO(-1)\oplus\cO(-1)\rightarrow \IP^1$, one has $N=1$,
$B=\{\pm \beta\}$ where  $\beta$ is the class of $\IP^1$, and
\be
\label{Omconifold}
\chi_X=2\, ,
\qquad
n_{\beta}=n_{-\beta} = 1.
\ee

Note that for such an uncoupled BPS structure, the factor $\Omega(\gamma') \langle\gamma,\gamma'\rangle$
in \eqref{eqTBA} always vanishes for $\gamma\in\Gamma_e$. This implies that for `electric' charges $\gamma_e\in\Gamma_e$
the solution coincides with its `semi-flat' limit, i.e. $\cX_{\gamma_e}=\cXsf_{\gamma_e}$. The non-trivial task is therefore to
determine $\cX_{\gamma}$ for `magnetic charges' $\gamma\in\Gamma_m$. In view of \eqref{qtorus},
it suffices to determine the holomorphic Fourier
modes $\cX_{\delta^\vee}$ and $\cX_{\alpha_i^\vee}$ associated to the duals of  $\delta$ and $\alpha^i$  in $\Gamma_m$.
Accordingly, we define
\be
\cY_0=\log\(\cX_{\delta^\vee}/\cXsf_{\delta^\vee}\),
\qquad
\cY_i=\log\(\cX_{\alpha_i^\vee}/\cXsf_{\alpha_i^\vee}\).
\label{defcY}
\ee
which are functions of  $z=(w,v^1,\dots, v^{N-1})$, $\vth=(\vph,\vth^1,\dots,\vth^{N-1})$ and $t$.

Thus, we arrive at the following RH problem:
\begin{prob}
\label{probRH}
Let us denote
$p_\beta(t)=e^{2\pi\I(\thb-\vbt/t)}$ and $q(t)=e^{2\pi\I(\vph-w/t)}$.
Find functions $\cY_\mu(z,\vth,t)$ ($\mu=0,1,\dots, N-1$) satisfying the following
conditions
\begin{enumerate}
\item
For all $\mu$, $\exp(\cY_\mu)$
is meromorphic in the domain\footnote{In contrast to the RH problems in \cite{Bridgeland:2017vbr,Bridgeland:2019fbi},
we relax the meromorphy condition by excluding the locus $q=1$.
This weaker condition is necessary in order that solutions exist at all, in fact
the solution found in \cite{Bridgeland:2019fbi} fails to satisfy
the stronger meromorphy condition for non-vanishing fiber coordinates.}
\be
\begin{split}
t\in&\, \IC\setminus\Bigl(\ell_{\delta}\cup\ell_{-\delta}\cup(\cup_{\beta\in B}\cup_{k\in \IZ} \ell_{\beta+k\delta})\cup \{q(t)=1\}\Bigr) ;
\end{split}
\nn
\ee
\item
$\cY_\mu$
are discontinuous across active BPS rays, with jumps given by
\be
\begin{array}{rlrl}
\mbox{across }\ell_{\beta+k\delta}:\quad
\Delta\cY_0=& -kn_\beta \log\(1-p_\beta\,q^k\),
\quad &
\Delta\cY_i=& -\beta_i  n_\beta \log\(1-p_\beta\, q^k\),
\\
\mbox{across }\ell_{\pm \delta}:\quad
\Delta\cY_0=& \mp \chi_X\sum_{k=1}^\infty 2k\log\(1-q^{\pm k}\),
\quad &
\Delta\cY_i=& 0;
\end{array}
\label{Pr-jumps}
\ee

\item Growth condition:
$\cY_\mu=o(1)$ as $t\to 0$ and $\cY_\mu=O(\log |t|)$ as $t\to \infty$;

\item
Reflection property:
\be
\cY_\mu(z,-\vth,-t)=-\cY_\mu(z,\vth,t).
\label{reflect-cY}
\ee

\end{enumerate}
\end{prob}

In the case of the conifold, this RH problem was solved  on
the `zero section' $\vth=0$ of the holomorphic tangent bundle  in \cite{Bridgeland:2017vbr}.
Its extension to general values of fiber coordinates was later determined in  \cite{Bridgeland:2019fbi}.
Our goal will be to revisit these results using the explicit
solution of the RH problem for uncoupled BPS structures reviewed in
\S\ref{sec_RH}. The main problem with this solution is that the resulting
sums over charges turn out to be not absolutely convergent.
Hence, even if they can be defined, such definition may not be unique.
For example, they may depend on the order of different summations and integrations.
In the next two sections we present two different strategies to solve this problem.

\section{Naive resummation of the results for finite BPS structures}
\label{sec_pass1}

The simplest strategy to get a solution of the Riemann-Hilbert problem specified in the previous section,
i.e. to find the functions \eqref{defcY} and the associated $\tau$ function,
is to try to perform the sum over charges of the functions representing the solution in the finite case.
However, one immediately encounters the following problem. As apparent from \eqref{Xuncoupled}
and \eqref{taufun-res2}, all these functions depend on the fiber coordinates $\vth_\gamma$
through the bracket $[\vth_\gamma]$ defined in \eqref{defbr}, which ensures
 periodicity with respect to integral shifts of $\Re\vth_\gamma$.
This is just a manifestation of the fact that they follow from the integral formula \eqref{eqTBA}
which is manifestly invariant under such shifts. As a result, the sum over charges
will produce a sum over $k\in\IZ$ of functions depending
on $[\thb+k\vph]$. Such infinite sums are discontinuous at every
$\Re\vph\in \IQ$, which is of course not acceptable. To avoid this problem,
we {\it postulate} that the brackets around $\vph$ should be ignored,
namely we set by hand $[\thb+k\vph]=[\thb]+k\vph$.
This can be interpreted as the result of an analytic continuation from $\Re\vph=0$.\footnote{To avoid cluttering,
we will also omit the brackets on $\thb$ concentrating on the domain where they belong to the interval $(0,1)$.}

Having dropped the brackets, we use the Binet-like integral representations for the generalized Gamma and Barnes functions
given in Appendix \ref{ap-functions}, exchange the sum over charges with the integral and perform the sum over $k.$
Below we summarize the main results of this procedure,
relegating details to Appendix \ref{ap-conifold}.

\begin{itemize}

\item
For $\vth_\gamma=0$, in the case of the conifold, this procedure applied to \eqref{Xuncoupled} and \eqref{taufun-res2}
leads to a set of integrals along the real axis, with an integrand having a pole at the origin
and defined through the principal value prescription.
The resulting Darboux coordinates and $\tau$ function turn out to coincide with the
solution of the Riemann-Hilbert problem constructed in \cite{Bridgeland:2017vbr}.

\item
For more general small crepant resolutions and vanishing $\vth_\gamma$, the pole in the integrals expressing $\cY_0$ becomes
of higher order, unless
\be
\sum_{\beta\in B^+} n_\beta=\hf\,\chi_X,
\label{condchi}
\ee
which holds coincidentally for the conifold.\footnote{For the generalized conifold with
 $(N_0,N_1)=(2,1)$, also known as  the suspended pinch point singularity,
the GV invariants take values $\{1,1,-1\}$ for the three rational curves, while $\chi_X=3$
(see example 5.14 in \cite{Mozgovoy:2020has}).}
As a result, the integrals cannot be defined by the principal value prescription anymore.
Nevertheless, if one {\it ignores} this divergence by shifting the integration contour
away from the pole by hand, one arrives at a solution which is a superposition of the functions
appearing in the solution for the conifold, as was suggested in \cite{Bridgeland:2019fbi}.

\item
For $\vth_\gamma\ne 0$,
after performing the summation over $k$, one ends up with divergent integrals even in the conifold case.
If, as above, one {\it ignores} this divergence by changing the integration contour
(and drops a few other divergent contributions),
it is possible to bring the functions introduced in \eqref{defcY} to the following form
\bea
\cY_i\Br &=& \sum_{\beta \in B^+}\beta_i \,n_\beta\[ \cF_{0,0}(\vbt-t\thb|w-t\vph, -t)
-\(\frac{\vbt\vph-w\thb}{w}+\hf\)\log\(1-e^{2\pi\I \vbt/w}\)
\right.
\nn\\
&&\left.
+\frac{w-t\vph}{2\pi\I t}\, \Li_2\(e^{2\pi\I  \vbt/w}\)\],
\label{solBridge}\\
\cY_0\Br &=&
\sum_{\beta \in B^+\cup\{0\} }n_\beta\[ \cF_{1,0}(\vbt-t\thb|w-t\vph, -t)
+\frac{\vbt}{w}\(\frac{\vbt\vph- w\thb}{w}+\hf\)\log\(1-e^{2\pi\I \vbt/w}\)
\right.
\nn\\
&&\left.
+\(\frac{1}{\pi\I}\(\frac{\vbt\vph-w\thb}{w}+\frac14\)-\frac{\vbt-t\thb}{2\pi\I t}\) \Li_2\(e^{2\pi\I  \vbt/w}\)
-\frac{w-t\vph}{2\pi^2 t}\, \Li_3\(e^{2\pi\I  \vbt/w}\)\],
\nn
\eea
where the functions $\cF_{n,m}$ are defined in \eqref{deffunF} and we recall that $n_0=-\hf\, \chi_X$.
In the conifold case where the GV invariants are given by \eqref{Omconifold},
this solution coincides with the one in \cite{Bridgeland:2019fbi}.
This is the reason why we labelled it by the superscript $\Br$ to distinguish it from another solution discussed in the next section.

\item
The $\tau$ function shows a similar behavior. If the condition \eqref{condchi} does not hold,
or  if $\vth_\gamma$ are non-zero,
it also has a divergence. (For $\vth_\gamma=0$, the divergence can be traced back to the last logarithmic term in \eqref{restau}.)
Proceeding then as for Darboux coordinates by shifting the contour and dropping some terms, it turns out to be possible to get
the $\tau$ function corresponding to the solution \eqref{solBridge},
\be
\begin{split}
\log\tau\Br =&\,
\sum_{\beta \in B^+\cup\{0\} }n_\beta\biggl[
\cF_{0,1}(\vbt-t\thb|w-t\vph,-t)+\vph\,\cF_{1,0}(\vbt-t\thb|w-t\vph, -t)
\biggr.
\\
&\,\biggl.
+\thb\, \cF_{0,0}(\vbt-t\thb|w-t\vph, -t)
+\frac{w^2}{(2\pi\I t)^2}\,\Li_3\(e^{2\pi\I \vbt/w}\)\biggr].
\end{split}
\label{taufun-Bridge}
\ee
It generalizes the $\tau$ function found in \cite{Bridgeland:2017vbr} to $\vth_\gamma\ne 0$ and beyond the conifold case.

\item
For the Plebanski potential, however, the sum in \eqref{Wfinite} does not make any sense and does not allow to reproduce the result
found in \cite{Bridgeland:2019fbi}.

\end{itemize}

Thus, we conclude that the approach based on the summation of the results
for finite BPS structures requires {\it ad hoc} manipulations
with divergent integrals to produce a well defined solution.
Such manipulations involve many ambiguities and cannot be justified.
Furthermore, this approach completely fails for the Plebanski potential.

\section{New solutions of the RH problem}
\label{sec_pass2}

Having shown the difficulties in summing up the contributions from the infinite trajectories of active BPS charges directly,
we shall now explore a different strategy, which uses the formulation of the conformal TBA equations in terms
of rational DT invariants as a starting point.
We propose a prescription for making sense of the sums over active charges,
which however has an inherent ambiguity parametrized by $\eps=\{0,\pm\}$.
This gives rise to new solutions to the RH problem for the conifold and other small crepant resolutions,
which we analyze and compare.

\subsection{Strategy}
\label{subsec-strategy}

An important step in arriving at the results \eqref{Xuncoupled} and \eqref{taufun-res2} relevant
for finite BPS structures was to combine the contributions of opposite charges $\pm\gamma$
\cite{Gaiotto:2014bza,Barbieri:2018swu,Alexandrov:2021wxu}. However, in the presence of
an  infinite BPS spectrum and for fiber variables
with a non-vanishing imaginary part, this recombination might be not innocuous.
Indeed, let us consider as an example (relevant for small crepant resolutions)  a `trajectory'
of active charges $\gamma=\beta+k\delta$ with the same DT invariant. Such a trajectory
typically produces contributions of the form
\be
\sum_k (\,\cdots) \, n_\beta\,e^{2\pi\I (\thb+k\vph)}
\ee
to the Darboux coordinates or to the $\tau$ function, where the factor in bracket
grows or decays polynomially
as a function of $k$. Such a series is only convergent provided $\Im\vph$
has a definite sign (positive if the sum runs over positive $k$). On the other hand, there is a similar contribution
from opposite charges where the sign in the exponential is flipped and hence convergence requires
the opposite sign of $\Im\vph$. Thus, these two incompatible convergence requirements
prevent us from combining the two contributions.

Note that due to the problem explained in the beginning of \S\ref{sec_pass1},
we do not have the option of setting $\Im\vph=0$ and then analytically continuing to complex values, as this would lead
to functions which are discontinuous at every $\Re\vph\in \IQ$. Instead,  we should rather set
$\Re\vph=0$ and then analytically continue away from the imaginary axis.

This consideration suggests the following strategy.
Rather than combining the contributions of the opposite charges {\it before} summing along each trajectory
in the charge lattice, one should first perform each sum separately.
To make these sums feasible, it is furthermore convenient to work in terms of rational DT invariants \eqref{defbOm}.
Although the rational invariants are supported on a larger set of charge vectors
(in particular, it is infinite even for a finite BPS structure),
most of expressions considerably simplify when written in terms of  $\bOm(\gamma)$.
Then, as we shall see below, all relevant quantities can be expressed through the following sum
\be
S(v,w,\vph)=\sum_{k\in\IZ}^\infty \frac{ e^{2\pi\I k \vph}}{v+k w}.
\label{Seriesk}
\ee
This sum is still problematic because it involves both positive and negative $k$ and hence it is divergent for any $\Im\vph\ne 0$.
To make sense out of it, we {\it define} it to be the function
\be
S'(v,w,\vph,\vph')=\sum_{k=0}^\infty \frac{ e^{2\pi\I k \vph}}{v+wk}
+\sum_{k=1}^\infty \frac{ e^{2\pi\I k \vph'}}{v-wk}\, ,
\label{seriesk2}
\ee
which is well defined for $\Im\vph,\Im\vph'>0$, analytically continued to $\vph'=-\vph$.

For example, if one applies this prescription to $S(1,0,\vph)$, one obtains
\be
\sum_{k\in\IZ}e^{2\pi \I k \vph}
=\[\frac{ 1}{1- e^{2\pi\I\vph}}+\frac{ e^{2\pi \I \vph'}}{1- e^{2\pi\I \vph'}}\]_{\vph'=-\vph}=0,
\label{sumcos-der}
\ee
rather than the usual Dirac comb distribution.
In the general case, the sum can be evaluated using the obvious relation
\be
\sum_{k=0}^\infty\frac{x^{k+a}}{k+a}=\int_0^x \frac{\de y}{y} \, \frac{y^a}{1-y},
\qquad
|x|<1, \ \Re a>0.
\label{sumint}
\ee
Assuming for concreteness\footnote{Note that this assumption can easily be dropped.
Indeed, if it does not hold, one shifts the summation variable in
\eqref{Seriesk} as $k\to k-n$ with $n=\fl{\Re(v/w)}$, so that one gets
$$
S(v,w,\vph)=e^{-2\pi\I n\vph }S(v',w,\vph),
$$
where $v'=v-nw$. The resulting exponential factor is combined with a similar factor in \eqref{Series4} to produce the full $v$,
whereas in all three possible results in \eqref{possibleInt}
$v'$ can be replaced by $v$ due to their invariance under integral shifts of $v/w$.
Thus, one arrives at the same result.}
that $0<\Re(v/w)<1$ and denoting $x=e^{2\pi\I \vph}$, $x'=e^{2\pi\I \vph'}$,
one then finds
\be
S'(v,w,\vph,\vph')=
\frac{e^{-2\pi\I \vph v/w}}{w}\[\int_0^x \frac{\de y}{y} \, \frac{y^{v/w}}{1-y}
+(xx')^{v/w}\int_{1/x'}^\infty \frac{\de y}{y} \, \frac{y^{v/w}}{1-y}\],
\label{seriesk3}
\ee
where in the second integral we inverted the integration variable.
Making the analytic continuation to $x'=1/x$, one arrives at
\be
S(v,w,\vph)=\frac{e^{-2\pi\I  \vph v/w}}{w}\int_0^\infty \frac{\de y}{y} \, \frac{y^{v/w}}{1-y}\, .
\label{Series4}
\ee
The problem now boils down to the fact that there is a pole at $y=1$ and one should specify how it is treated.
Effectively, this encodes all the ambiguity of the solution to our RH problem.
There are three natural possibilities:
\be
\int_0^\infty \frac{\de y}{y} \, \frac{y^{v/w}}{1-y}=
\begin{cases}
\pi \, \cot(\pi v/w) \quad & \mbox{principal value prescription,}
\\
\frac{\pi\, e^{\pi \I v/w}}{\sin(\pi v/w)} \quad &  \mbox{contour is above the pole,}
\\
\frac{\pi\, e^{-\pi \I v/w}}{\sin(\pi v/w)} \quad & \mbox{contour is below the pole.}
\end{cases}
\label{possibleInt}
\ee
We shall label  these three possible prescriptions by $\eps=0$ and $\pm$, respectively.
They lead  to three possible definitions of the function \eqref{Seriesk}
\be
S^{(\eps)}(v,w,\vph)=\frac{\pi}{w}\, e^{-2\pi\I  \vph v/w}\Bigl[\cot(\pi v/w)+\I\eps \Bigr].
\label{Seps}
\ee

One can argue that the principle value prescription is preferred for the following reasons.
First, the function \eqref{Seriesk} has the obvious symmetry $S(v,-w,-\vph)=S(v,w,\vph)$,
which is realized by \eqref{Seps} only for $\eps=0$. Second, at $\vph=0$ the sum in \eqref{Seriesk} becomes absolutely convergent
after combining contributions of $k$ and $-k$ and leads to the unambiguous result
\be
S(v,w,0)=\frac{1}{v}-\frac{2v}{w^2}\sum_{k=1}^\infty \frac{1}{k^2-(v/w)^2}=\frac{\pi}{w}\, \cot(\pi v/w),
\label{evalS0}
\ee
where we used \eqref{sumcos0} at $x=0$, which again agrees with \eqref{Seps} only for $\eps=0$.
Nevertheless, we shall entertain the possibility of choosing $\eps\ne 0$, as
it will allow us to expose possible ambiguities and make contact with results in the literature.

Below we express the Darboux coordinates, $\tau$ function and Plebanski potential
for arbitrary small crepant resolution in terms of the function $S$ and, replacing it by $S^{(\eps)}$ \eqref{Seps},
obtain explicit form of possible solutions to the RH problem. Then in \S\ref{subsec-compare}
we will analyze their properties and relations.

\subsection{Darboux coordinates}
\label{Dc-new}

Substituting \eqref{eqTBA} and the BPS data \eqref{OmCY} into the functions \eqref{defcY},
one finds that they can be written in the following form
\be
\cY_i=\sum_{\beta\in B} \tn_\beta \,\beta_i \,\cYa,
\qquad
\cY_0=\sum_{\beta\in B\cup\{0\}} \tn_\beta\, \cYb,
\label{cY-1}
\ee
where we introduced\footnote{Recall that $n_0$ was  defined below \eqref{solBridge}  as $-\hf\, \chi_X$, here we need twice this value.
Note also that the functions $\cYa$, $\cYb$ defined here should not be confused
the functions defined in \eqref{cYabc}.
Here they include the contribution of only one $\beta\in B$, whereas in
Appendix \ref{ap-conifold}
they combine the contributions of $\beta$ and $-\beta$.
We hope that this abuse of notation will not lead to any confusion.}
\be
\tn_\beta=\left\{ \begin{array}{ll}
n_\beta, & \qquad \beta\in B,
\\
-\chi_X, & \qquad \beta=0,
\end{array}\right.
\label{deftnb}
\ee
and
\be
\begin{split}
\cYa =&\, -\sum_{n=1}^\infty \frac{1}{2\pi\I n}\sum_{k\in\IZ}
\int_{\ell_{\beta+k\delta}}\frac{\de t'}{t'}\, \frac{t\, \cXsf_{n(\beta+k\delta)}(t')}{t'-t}\, ,
\\
\cYb =&\, -\sum_{n=1}^\infty \frac{1}{2\pi\I n}\sum_{k\in\IZ} k
\int_{\ell_{\beta+k\delta}}\frac{\de t'}{t'}\, \frac{t\, \cXsf_{n(\beta+k\delta)}(t')}{t'-t}\, .
\end{split}
\ee
Changing the integration variable as $t'=\I(\vbt+kw)/s$ and denoting $\hvb=\vbt/t$, $\hw=w/t$,
these functions are expressed through the function \eqref{Seriesk} as follows
\be
\begin{split}
\cYa =&\, \sum_{n=1}^\infty \frac{1}{2\pi n}\int_0^\infty \de s\,e^{-2\pi n (s-\I\thb)} S(\hvb+\I s,\hw,n\vph)\, ,
\\
\cYb =&\, \sum_{n=1}^\infty \frac{1}{2\pi n \hw}\int_0^\infty \de s \,e^{-2\pi n (s-\I\thb)}\sum_{k\in\IZ}
\Bigl(e^{2\pi \I n k\vph}-(\hvb+\I s)\,S(\hvb+\I s,\hw,n\vph)\Bigr)\, .
\end{split}
\ee
Then, using \eqref{sumcos-der} and \eqref{Seps}, and performing the sum over $n$, they can be expressed
as the following integrals
\bea
\cYai{\eps} &=&  \frac{1}{2\I\hw}\int_0^\infty \de s\,\Li_1\(e^{-2\pi (1-\vph/\hw)s+2\pi\I (\thb -\vph\hvb/\hw)}\)
\Bigl(\coth(\pi(s-\I \hvb)/\hw)-\eps\Bigr) ,
\\
\cYbi{\eps}&=& -\frac{1}{2\hw^2}\int_0^\infty \de s
\, (s-\I \hvb)\,
\Li_1\(e^{-2\pi (1-\vph/\hw)s+2\pi\I (\thb -\vph\hvb/\hw)}\)
\Bigl(\coth(\pi(s-\I \hvb)/\hw)-\eps\Bigr),
\nn
\eea
where $\eps$ indicates the ambiguity in the function $S$ which we left unresolved.
Finally, the first terms in these integrals can be rewritten in terms of the functions
\eqref{deffunJ}, whereas the second terms can be explicitly integrated:
\be
\begin{split}
\cYai{\eps} =&\, \frac{1}{2\I \hw}\,\cJ_{1,0}(\hw,\hvb,\vph,\thb)
+\frac{\I \eps\,\Li_2\( e^{2\pi\I (\thb -\vph\hvb/\hw)}\)}{4\pi(\hw-\vph)}\,,
\\
\cYbi{\eps}=&\, \frac{1}{2\hw^2}\Bigl(\cJ_{1,1}(\hw,\hvb,\vph,\thb)
+\I \hvb \cJ_{1,0}(\hw,\hvb,\vph,\thb) \Bigr)
\\
&\,
+\frac{\eps\,\Li_3\( e^{2\pi\I (\thb -\vph\hvb/\hw)}\)}{8\pi^2(\hw-\vph)^2}
+\frac{\eps\hvb\,\Li_2\( e^{2\pi\I (\thb -\vph\hvb/\hw)}\)}{4\pi\I\hw (\hw-\vph)}\, .
\end{split}
\label{rescY}
\ee

Of course, the resulting Darboux coordinates
by construction satisfy the necessary properties of a solution to the RH problem,
such as the right KS transformations across BPS rays, the reflection property, etc.
Nevertheless, it is instructive to check that they indeed hold. For example, the KS transformations \eqref{Pr-jumps} follow
from Propositions \ref{porp-jump} and \ref{porp-jump0}\footnote{The KS transformations across $\ell_{\pm \delta}$ are generated only
by the contribution from $\beta=0$. In this case, the discontinuities are captured by the first term in
the second equation of \eqref{rescY} which is proportional to $\cJ_{1,1}(\hw,0,\vph,0)$. This function differs from
the function $\tcJ_{1,1}(\hw,\vph)$ defined in \eqref{deffunJ0} by a regular term, hence
its discontinuities are described by Proposition \ref{porp-jump0}.} in Appendix \ref{ap-new} describing the jumps of the functions $\cJ_{m,n}$
across hypersurfaces which, upon substitution of the arguments as in \eqref{rescY}, precisely coincide with the BPS rays
\eqref{BPSray} for active charge vectors.

\subsection{$\tau$ function}
\label{subsec-taufunnew}

Our starting point is the formula \eqref{restau}. For the BPS structure \eqref{OmCY}, it takes the form
\be
\log\tau =  \sum_{\beta\in B\cup \{0\}}\tn_\beta\, (\cT_\beta+\cL_\beta),
\label{tau-full1}
\ee
where
\bea
\cT_\beta &=&
\sum_{n=1}^\infty \frac{1}{4\pi^2 n^2}
\!\!\! \sum_{\beta\ne 0:\ k\in\IZ\atop \beta=0:\ k\in \IZ\setminus\{0\}} \!\!\!
e^{2\pi\I n (\thb+k\vph)}
\int\limits_{\ell_{\beta+k\delta}} \frac{\de t'}{t'} \,\frac{t}{t'-t}
\(1+2\pi\I n\, \frac{\vbt+k w}{t'}\)e^{-2\pi\I\, \frac{\vbt+kw}{t'}},
\label{cTbeta}
\\
\cL_\beta &=& -\frac{1}{4\pi^2}\sum_{n=1}^\infty \frac{1}{n^2}
\!\!\! \sum_{\beta\ne 0:\ k\in\IZ\atop \beta=0:\ k\in \IZ\setminus\{0\}} \!\!\!
e^{2\pi\I n (\thb+k\vph)}\log\bigl(n (\hvb+k\hw)\bigr).
\label{cLbeta}
\eea
Note that for vanishing $\beta$, the value $k=0$ is excluded from the sum.

First, let us evaluate the contribution $\cT_\beta$. Changing the integration variable as in the previous subsection,
it can also be expressed through the function \eqref{Seriesk}, which gives
\be
\cT_\beta =
\sum_{n=1}^\infty \frac{1}{4\I\pi^2 n^2}\int_0^\infty \de s\, (1+2\pi n s)\,e^{-2\pi n (s-\I\thb)}
\(S(\hvb+\I s,\hw,n\vph)-\frac{\delta_{\beta=0}}{\I s}\),
\label{cTbeta2}
\ee
where the last term originates from the absence of one term in the sum.
Then, proceeding as for Darboux coordinates, we substitute \eqref{Seps} and perform the sum over $n$.
For $\beta\ne 0$, the result can again be expressed through the functions \eqref{deffunJ} and reads
\be
\cT^{(\eps)}_\beta=
\frac{1}{4\pi \hw}\,\Bigl(\cJ_{2,0}(\hw,\hvb,\vph,\thb)
+2\pi \cJ_{1,1}(\hw,\hvb,\vph,\thb)\Bigr)
+\frac{\eps(2\hw-\vph)}{8\pi^2(\hw-\vph)^2}\,\Li_3\( e^{2\pi\I (\thb -\vph\hvb/\hw)}\).
\label{res-cTbeta}
\ee
On the other hand, for $\beta=0$ this cannot be done because the function $\cJ_{2,0}(x,y,\xi,\eta)$ has a pole at $y=0$.
In \eqref{cTbeta2} this pole is cancelled by the last term. Thus, instead of the functions $\cJ_{m,n}$,
the contribution $\cT_0$ can be written in terms of their modified version $\tcJ_{m,n}$ defined in \eqref{deffunJ0}.
As a result, one obtains
\be
\cT^{(\eps)}_0=
\frac{1}{4\pi \hw}\,\Bigl(\tcJ_{2,0}(\hw,\vph)+2\pi \tcJ_{1,1}(\hw,\vph)\Bigr)
+\frac{\eps(2\hw-\vph)\zeta(3)}{8\pi^2(\hw-\vph)^2}\, .
\label{res-cT0}
\ee

Next, we turn to the contribution \eqref{cLbeta}. As  written, it is clearly divergent.
However, recall that the $\tau$-function defined by \eqref{eqtau} is ambiguous up to terms independent of $\vbt$ and $w$.
Therefore, we can safely ignore them in our calculations, and we will use the sign $\simeq$
to denote equalities valid up to such terms.
In fact, it turns out that all divergences in $\cL_\beta$ come from such irrelevant contributions,
whereas the relevant part of this function is convergent.
Indeed, for $\beta=0$ we have
\be
\cL_0\simeq \frac{1}{4\pi^2}\sum_{n=1}^\infty \frac{1}{n^2}\, \log(\hw)=\frac{1}{24}\,\log(\hw) ,
\label{res-cL0}
\ee
whereas for non-vanishing $\beta$, one obtains an integral of the function \eqref{Seriesk}:
\be
\cL_\beta \simeq  -\frac{1}{4\pi^2}\sum_{n=1}^\infty \frac{e^{2\pi\I n \thb}}{n^2}\[
\int_0^{\hvb} \de y \(S(y,\hw,n\vph)-\frac{1}{y}\)
+\log(\hvb/\hw) \].
\ee
Substituting \eqref{Seps} and performing the sum over $n$, the resulting integral can be expressed through another
special function defined in \eqref{def-cIn}, which gives
\be
\cL^{(\eps)}_\beta \simeq
-\frac{1}{4\pi^2}\, \cI_2(\vbt/w,\vph,\thb)
+\frac{\eps}{8\pi^2\vph}\,\Bigl(\Li_3\(e^{2\pi\I (\thb-\vph \hvb/\hw)}\)- \Li_3\(e^{2\pi\I \thb}\)\Bigr)\, .
\label{res-cLbeta}
\ee
Note that the last term is actually regular at $\vph=0$. For $\eps=0$, plugging \eqref{res-cTbeta}, \eqref{res-cT0},
\eqref{res-cL0} and \eqref{res-cLbeta} into \eqref{tau-full1}, one recovers the formula \eqref{reztau-pv} for the $\tau$ function
given in the Introduction.

It is straightforward to check that for any $\eps$ the resulting $\tau$ function is consistent with the Darboux coordinates
found in the previous subsection. Indeed, in our case the relation \eqref{eqtau} is equivalent to the conditions
\be
\begin{split}
\frac{\p}{\p \vbt}\(\cT^{(\eps)}_\beta+\cL^{(\eps)}_\beta\)=&\,-t\,\frac{\p\cYai{\eps}}{\p t},
\\
\frac{\p}{\p w}\(\cT^{(\eps)}_\beta+\cL^{(\eps)}_\beta\)=&\,-t\,\frac{\p\cYbi{\eps}}{\p t},
\end{split}
\ee
which are easily verified by noticing that
the Darboux coordinates are annihilated by the operator  $t\p_t+\hw\p_{\hw}+\hvb\p_{\hvb}$,
and by using the relations from Proposition \ref{prop-ident} and the first relation in \eqref{ident-more}.

\subsection{Plebanski potential}

The Plebanski potential $W$ determining the metric on the holomorphic tangent bundle for
an uncoupled BPS structure is given by the formula \eqref{resW}.
In our case it takes the form
\be
W=-\chi_X W_0+\sum_{\beta\in B}n_\beta\, W_\beta,
\label{Wtoric}
\ee
where
\bea
W_0 &=&\frac{1}{(2\pi\I)^2 w}\, \sum_{n=1}^\infty \frac{1}{n^3}\sum_{k\ne 0} \frac{e^{2\pi n k\vph}}{k}\, ,
\label{W0}\\
W_\beta &=& \frac{1}{(2\pi\I)^2}\, \sum_{n=1}^\infty \frac{e^{2\pi\I n \thb}}{n^3} S(\vbt,w,n\vph) \, .
\label{Wbeta}
\eea
We define the first function following the same strategy as in \S\ref{subsec-strategy}.
Then the sum over $k$ produces a difference of two logarithms which can be recombined (see \eqref{Li-ident})
leaving behind a linear function of $\vph$.\footnote{This calculation can also be done using the function $S$. Indeed, one has
$$
\sum_{k\ne 0} \frac{e^{2\pi n k\vph}}{k}=\lim_{v\to 0}\(S(v,1,\vph)-\frac{1}{v}\)=-2\pi\I\vph+\pi\I \eps.
$$
}
However, the Plebanski potential is defined only up to a linear function
in  fiber coordinates because only its second derivatives with respect to $\theta^a$
appear in the metric components, the section condition \eqref{actonX1}
and other geometric relations \cite{Bridgeland:2019fbi,Bridgeland:2020zjh}.
Therefore, we can drop the first term proportional to $W_0$ in \eqref{Wtoric}.
The remaining contribution is obtained by substituting \eqref{Seps}, which gives
\be
W_\beta^{(\eps)}=-\frac{1}{4\pi w}\,
\Li_3\(e^{2\pi\I(\thb-\vph \vbt/w)}\)\Bigl( \cot(\pi \vbt/w)+\I \eps\Bigr).
\label{resWbeta}
\ee
Note that the first term in the brackets in \eqref{resWbeta} is an odd function of $\beta$.
Therefore, its contribution to $W$ can be
further simplified by recombining $W_\beta$ and $W_{-\beta}$ using the identity \eqref{Li-ident}
which allows to replace $\Li_3$ by the Bernoulli polynomial $B_3$.
As a result, for $\eps=0$ one finds exactly the same function which is presented in \eqref{resW-pv}.

It is straightforward to check that for any $\eps$ the resulting Plebanski potential
satisfies the section condition \eqref{actonX1}.
Indeed, in our case it is equivalent to the conditions
\be
\begin{split}
&
\(\frac{\p}{\p\hvb}+\frac{\p}{\p\thb}\)\cYai{\eps}
=-\frac{t}{2\pi\I}\, \frac{\p^2 W_\beta^{(\eps)}}{\p\thb^2},
\qquad
\(\frac{\p}{\p\hw}+\frac{\p}{\p\vph}\)\cYbi{\eps}
=-\frac{t}{2\pi\I}\, \frac{\p^2 W_\beta^{(\eps)}}{\p\vph^2},
\\
&\qquad\qquad
\(\frac{\p}{\p\hw}+\frac{\p}{\p\vph}\)\cYai{\eps}
=
\(\frac{\p}{\p\hvb}+\frac{\p}{\p\thb}\)\cYbi{\eps}
=-\frac{t}{2\pi\I}\, \frac{\p^2 W_\beta^{(\eps)}}{\p\thb\p\vph},
\end{split}
\ee
which are easily verified using the last two relations in \eqref{ident-more}.

\subsection{Comparison of different solutions}
\label{subsec-compare}

Let us now analyze and compare the solutions found in the previous subsections.
First of all, one can note that they differ by quite simple terms given by polylogarithms (see \eqref{rescY}).
Thus, in this approach the ambiguity is rather mild. Nevertheless, it is desirable to fix it somehow.

To this end, let us first consider  the analytic properties of these solutions.
In Problem \ref{probRH} we required the functions $\exp(\cY_\mu)$
to be meromorphic away from the BPS rays and the locus $e^{2\pi\I (\vph-\hw)}=1$,
which was not excluded in the RH problem of \cite{Bridgeland:2017vbr,Bridgeland:2019fbi}.
The reason for this modification is due to the existence of a  pole at $\xi=x$
in the functions $\cJ_{m,n}(x,y,\xi,\eta)$ defined and analyzed in  Appendix \ref{ap-new}.
The pole leads to an essential singularity at $t=w/\vph$ in the Darboux coordinates $\cX_\gamma$.
It is natural to ask if requiring the absence of singularity at this point could resolve the ambiguity.
The singularity is certainly present for the solution with $\eps=0$.
For $\eps=\pm$, comparing \eqref{rescY} with the representation \eqref{deffunJ-an} of $\cJ_{m,n}$
one observes that the $\eps$-dependent terms exactly cancel the poles of the functions $\cJ_{m,n}$
provided one can identify $\eps$ with $\veps=\sgn\Re(w/t)$.
Unfortunately, the choice $\eps=\veps$ appears to be inconsistent.
Indeed, $\veps$ takes different values in the two half-planes of the $t$-complex plane
separated by the BPS rays $\ell_\delta\cup\ell_{-\delta}$. Therefore, identifying $\eps=\veps$, we take the solution with $\eps=+$
in one half-plane and with $\eps=-$ in the other. But this would spoil the KS transformations \eqref{Pr-jumps}
across the BPS rays $\ell_{\pm\delta}$. Hence, the singularity is always present, at least in one of the half-planes.

Let us now address the relation to the solution found in \cite{Bridgeland:2017vbr,Bridgeland:2019fbi}.
At first sight, all our results for Darboux coordinates, $\tau$ functions and
Plebanski potential look quite different from those appearing in these references. We shall now argue that they actually
agree for a suitable choice of prescription $\eps\in\{0,\pm\}$,
which however must be chosen independently for each $\beta\in B\cup \{0\}$.
While such a dependence may at first sound unnatural, it is in fact acceptable
since D2-branes wrapped on different rational curves have vanishing antisymmetric pairing.
More precisely, let us split the active charges into three classes
of trajectories, $\beta\in B^\pm$ and $\beta=0$, and choose independent
prescriptions $(\eps_+,\eps_-,\eps_0)$ for each of them.
While such a choice does not allow to improve any properties of the solutions
(in particular, it does not allow to remove the essential singularity),
 it does allow to recombine the contributions of $\pm \beta$ in the Plebanski potential \eqref{Wtoric}. Indeed,
for $\eps_\pm=\pm$, one finds that (recall that $W_0$ can be dropped)
\be
W\Br\equiv \sum_{\beta\in B^+}n_\beta\, \(W^{(+)}_\beta+W^{(-)}_{-\beta}\)
=\frac{2\pi^2}{3 w}\sum_{\beta\in B^+}n_\beta \, \frac{ B_3\(\[\thb-\vph\, \frac{\vbt}{w}\]\)}{1-e^{-2\pi\I \vbt/w}}\, .
\label{WBr}
\ee
Provided one keeps only the relevant cubic term of the Bernoulli polynomial $B_3(x)=x^3+\dots$,
this is precisely the Plebanski potential
obtained for the conifold in \cite[\S9.5]{Bridgeland:2019fbi}.\footnote{In \cite{Bridgeland:2019fbi}
it was called `Joyce function' and was denoted by $J$.
Note that  the normalization of the coordinates $(z,\theta)$  in {\it loc. cit.}
differs from our conventions by a factor of $2\pi\I$; moreover the potential \eqref{WBr}
is odd in the fiber coordinates due to \eqref{symBn} and the property of the square brackets noticed below \eqref{defbr}.}

This observation suggests that the prescription $\eps_\beta=\sgn\Im ( \vbt/w)$ generates Bridgeland's solution.
Indeed, we claim that the functions given in \eqref{solBridge}
(which in the conifold case represent the solution in \cite{Bridgeland:2019fbi}) can also be written as
\be
\begin{split}
\cY_i\Br = &\, \sum_{\beta\in B^+} n_\beta \,\beta_i \(\cY^{(1,+)}_\beta-\cY^{(1,-)}_{-\beta}\),
\\
\cY_0\Br = &\, -\chi_X\cYci{0}+\sum_{\beta\in B^+} n_\beta\(\cY^{(2,+)}_\beta+\cY^{(2,-)}_{-\beta}\).
\end{split}
\label{cY-Br}
\ee
As a sanity check on this claim, it is instructive to inspect the behaviour of the two sides
of these relations near $\varphi=\hat w$.
On the right, it can be extracted from \eqref{rescY} and \eqref{deffunJ-an},
while on the left it follows from \eqref{solBridge}
and the asymptotic expansions of $\cF_{0,0}=\log F$ and $\cF_{1,0}=\log G$ given in \cite[Prop. 4.6]{Bridgeland:2017vbr}.
As a result, on both sides one finds the same singular behavior\footnote{This also demonstrates that,
contrary to the claim in \cite{Bridgeland:2019fbi}, $\cX_\gamma\Br$ are not meromorphic functions.
It is possible that this fact was overlooked because for $\vph=0$,
the case where \cite[Prop. 4.6]{Bridgeland:2017vbr} has been derived,
the singularity lies at $w/t=0$, which is outside of the domain of consideration, whereas for non-vanishing $\vph$
it moves inside the complex plane.}
\be
\begin{split}
\cY_i\Br\under{\sim}{\vph\to \hw}  &\, \sum_{\beta\in B^+} n_\beta \,\beta_i \,\frac{\Li_2\(e^{2\pi\I(\thb-\hvb)}\)}{2\pi\I (\vph-\hw)}\, ,
\\
\cY_0\Br\under{\sim}{\vph\to \hw} &\, \sum_{\beta\in B^+\cup\{0\}} n_\beta\, \frac{\Li_3\(e^{2\pi\I(\thb-\hvb)}\)}{4\pi^2 (\vph-\hw)^2}\, .
\end{split}
\label{cY-Br-poles}
\ee
In fact, for $\cY_i\Br$ the equivalence of the formulae \eqref{solBridge} and \eqref{cY-Br}
is a consequence of an integral representation of the Faddeev quantum dilogarithm due to Woronowicz
(see \eqref{WoroF}). For $\cY_0\Br$, we do not have a similar analytic proof. Instead, we have checked
the agreement of \eqref{solBridge} and \eqref{cY-Br} numerically for random values of the arguments.
This agreement implies a new integral representation of the triple sine function $S_3(z|\omega_1,\omega_1,\omega_2)$
with two identical arguments.
Furthermore, equating the two representations for the $\tau$ function, \eqref{taufun-Bridge} and the one following
from results of \S\ref{subsec-taufunnew} with $\eps_\beta=\sgn\Im ( \vbt/w)$, one obtains
yet another representation for $S_3$, and it is not obvious at all
 how this second integral representation is related to the first one (although both were checked numerically).
We state both representations as Conjecture \ref{conj-sine} in the end of Appendix \ref{ap-Br}.
Thus, assuming that this conjecture holds true, we conclude that Bridgeland's solution
follows from our formalism and has exactly
the same status as other solutions considered so far.

The identification of the solution from \cite{Bridgeland:2019fbi} among those of \S\ref{Dc-new} also allows us to identify the one
corresponding to the solution of \cite{Bridgeland:2017vbr} at vanishing fiber coordinates.
The point is that the solution of \cite{Bridgeland:2019fbi} does {\it not} reduce to the one of \cite{Bridgeland:2017vbr}
when the fiber coordinates are set to zero.
And it is easy to check that the difference between them is exactly the same as the difference
between \eqref{cY-Br} and the solution corresponding to the principle value prescription.
Thus, it is the latter that reproduces the standard solution at vanishing fiber coordinates.
This is also consistent with the observation made at the end of \S\ref{subsec-strategy} around \eqref{evalS0}
and confirms  that the solution with $\eps=0$, obtained by applying
the principle value prescription to the integral \eqref{possibleInt}, is the preferred one.

\section{Relation to topological string partition function}
\label{sec_top}

In this section, we discuss the relation between the $\tau$ function and the topological string
partition function $F_{\rm top}$.  This relation was first noted for the resolved conifold in
\cite{Bridgeland:2017vbr}, by computing the asymptotic expansion of $\log \tau(z,0,t)$ near $t=0$
for vanishing fiber coordinates, and observing that it coincides with the genus expansion of
$F_{\rm top}(z,\lambda)$ in powers of the topological string coupling $\lambda$, up to the genus 0 contribution.
While the mathematical or physical origin of this coincidence remains obscure,
we shall see that it extends to arbitrary small crepant resolutions.
Furthermore, we observe that  the full genus zero contribution is captured by the Plebanski potential,
which adds another layer to the mystery.

\subsection{Brief review of the topological string partition function}

The topological string partition function was defined in the seminal work \cite{Bershadsky:1993cx},
as the generating series $\cF_{\rm top}(z,\bar z,\lambda) = \sum_{g\geq 0} \lambda^{2g-2} \cF_g(z,\bar z)$
of the genus $g$ vacuum amplitudes  in A-twisted topological string theory on a Calabi-Yau threefold $X$.
The generating series has zero radius of convergence, but it is expected to be the asymptotic expansion
near $\lambda=0$ of a globally well-defined function of the K\"ahler moduli $z^i$ ($ i=1,\dots,b_2(X)$)
and of the topological string coupling $\lambda$. While  $\cF_g(z,\bar z)$ are not holomorphic in
the K\"ahler moduli,
they satisfy holomorphic anomaly equations, which allow to determine them from their holomorphic limit
$F_g(z)$, obtained by sending $\bar z^i\to \infty$ while keeping $z^i$ fixed. The holomorphic
topological string partition function is then defined by its asymptotic series
$F_{\rm hol}(z,\lambda)=\sum_{g\geq 0} \lambda^{2g-2} F_g(z)$, although it is no longer
globally defined on K\"ahler moduli space but only near the large volume limit.

In~\cite{Gopakumar:1998ii,Katz:1999xq}, it was argued that $F_{\rm hol}(z,\lambda)$
is entirely determined by the Gopakumar-Vafa invariants $n_\beta^g\in\IZ$, via \be
\begin{split}
\label{gvbps}
e^{F_{\rm hol}(z,\lambda)} =&\, e^{F_{\rm pol}(z,\lambda)}   [M(e^{-\lambda})]^{\chi_X/2} \,
\prod_{\beta\in H_2^+(X)} \prod_{k>0} \left(1-e^{-k\lambda+2\pi\I z_\beta}\right)^{k n_\beta^0}
\\
&\,\times \prod_{\beta\in H_2^+(X)} \prod_{g>0}
\prod_{\ell=0}^{2g-2}
\left(1-e^{-(g-\ell-1)\lambda+2\pi\I z_\beta}
\right)^{(-1)^{g+\ell} {\scriptsize \begin{pmatrix} 2g-2 \\ \ell \end{pmatrix}}
n_\beta^g}.
\end{split}
\ee
Here $F_{\rm pol}$ (known as the polar part) is given by
\be
\label{Fpol}
F_{\rm pol}(z,\lambda)=-\frac{(2\pi \I)^3}{6\lambda^2}\, \kappa_{ijk} z^i z^j z^k
-\frac{\pi\I}{12} \,c_{i} z^i,
\ee
where $z^i=\int_{\gamma^i} J$ are the K\"ahler moduli obtained by integrating the K\"ahler form $J$ on
a basis $\{\gamma^i\}_{i=1}^{b_2(X)}$ of $H_2(X,\IZ)$, $\kappa_{ijk}=\gamma_i\cap\gamma_j\cap \gamma_k$
the triple intersection product on the dual basis
$\{\gamma_i\}$  of $H_4(X,\IZ)$, $c_i=\int_{\gamma_i} c_2(TX) $,
and  $z_\beta=\int_\beta J$ for $\beta\in H_2(X,\IZ)$. The second factor involves a power of
the Mac-Mahon function $M(\lambda)=\prod_{n\geq 1}(1-e^{-n\lambda})^{-n}$, and can be combined with the
third factor at the expense of defining $n_\beta^0=-\chi_X/2$ for $\beta=0$ as in \eqref{solBridge}.
The asymptotic expansion of the Mac-Mahon function near $\lambda=0$
was computed in \cite[\S E]{Dabholkar:2005dt}\footnote{See also  \cite[\S 4.3]{Pioline:2006ni}
and \cite{Koshkin:2007mz} for a streamlined computation
using Mellin transform. In comparing these references it is useful to note that
$\frac{\zeta'(2)}{2\pi^2} -\frac{1}{12} ( \gamma_E +\log 2\pi )  = \zeta'(-1) - \frac1{12}$. }
\be
\log M(\lambda) =
\frac{\zeta(3)}{\lambda^2} +\frac{1}{12} \,\log\frac{\lambda}{2\pi} + \frac{\zeta'(2)}{2\pi^2}-\frac{\gamma_E}{12}
+ \sum_{n\geq 2}\frac{  B_{2n} B_{2n-2}  \lambda^{2n-2}  }{2n(2n-2) (2n-2)!}\, .
\ee
Upon setting
\be
\tilde F_{\rm pol}(z,\lambda) = F_{\rm pol}(z,\lambda)
+ \frac{\chi_X}{2} \[ \frac{1}{12}\, \log\frac{\lambda}{2\pi} + \frac{\zeta'(2)}{2\pi^2}-\frac{\gamma_E}{12} \],
\ee
and taking the logarithm in \eqref{gvbps}, we find that the holomorphic topological string partition function is given by
\bea
F_{\rm hol}(z,\lambda) = \tilde F_{\rm pol}(z,\lambda)
- \sum_{g=0}^{\infty} \sum_{\beta\in H_2^+(X) \cup \{0\}} n_\beta^g  \sum_{m=1}^{\infty}
\frac{1}{m} \left( 2\sinh \frac{m\lambda}{2} \right)^{2g-2} e^{2\pi\I m z_\beta}.
\eea
When the  only non-zero GV invariants $n_\beta^g$ are those with $g=0$
(in which case we drop the upper index as in the rest of this article),
which is the case for small crepant resolutions, the asymptotic expansion
for small $\lambda$ can be obtained by Taylor expanding
\be
\label{bernii}
\frac{1}{4\sinh^2(x/2)} = \frac{1}{x^2} - \sum_{n\geq 1}
\frac{ B_{2n}\, x^{2n-2}}{2n(2n-2)! }
\ee
and using $\zeta(3-2n) = - \frac{B_{2n-2}}{2n-2}$. This gives
\bea
F_{\rm hol}(z,\lambda) &=& \tilde F_{\rm pol}(z,\lambda)
- \sum_{\beta\in H_2^+(X) \cup \{0\}} n_\beta  \sum_{m=1}^{\infty}
\left[
\frac{1}{m^3\lambda^2}- \sum_{n\geq 1}
\frac{m^{2n-3} B_{2n}}{2n(2n-2)! }\, \lambda^{2n-2}
\right] e^{2\pi\I m z_\beta}
\nn\\
&=&\tilde F_{\rm pol}(z,\lambda) +\frac{\chi_X}{2} \left[
\frac{ \zeta(3)}{\lambda^2}
+ \sum_{n\geq 2} \frac{B_{2n} B_{2n-2}  \lambda^{2n-2}  }{2n(2n-2) (2n-2)! }\right]
\nn\\
&&- \sum_{\beta\in H_2^+(X)} n_\beta
\left[
\frac{\Li_3( e^{2\pi\I z_\beta})}{\lambda^2}- \sum_{n\geq 1}
\frac{ B_{2n}  \lambda^{2n-2} }{2n(2n-2)!}\,  \Li_{3-2n}( e^{2\pi\I z_\beta})
\right].
\label{Ftopexp}
\eea
Recalling that  $B_{2n}= (-1)^{n-1}\vert B_{2n}\vert$, the constant map contribution on the second line  is
proportional to the intersection product on the moduli space $\cM_g$ of curves of genus $g=n$ \cite{Faber:1998gsw}
\be
 \int_{\cM_g} c_{g-1}^3
= \frac{|B_{2g} B_{2g-2}|}{2g (2g-2) (2g-2)!} \, ,
\ee
while the coefficient of the non-constant maps on the third line is the Euler number \cite{harer1986euler},
\be
\chi(\cM_g) = \frac{|B_{2g}|}{2g(2g-2)!}\, .
\ee

\subsection{Comparison with the $\tau$ function}

The asymptotic expansion of the function $\tau(z,\theta,t)$ for vanishing fiber coordinates is easily
obtained from  \eqref{reztau-pv} by combining the contributions of opposite vectors
$\pm \beta$ in $B$ and using Proposition \ref{porp-expand} in Appendix \ref{ap-new}. Upon using
Euler's relation
\be
\zeta(2n) = \frac{(-1)^{n+1}\, (2\pi)^{2n}}{2\, (2n)!} B_{2n},
\ee
and setting $k=n-1$, we readily obtain
\bea
\log\tau(z,0,t) &=&
\sum_{\beta\in B^+} n_\beta\[
\sum_{n\geq 1}
\frac{ B_{2n} }{2n(2n-2)!}\,  \Li_{3-2n}( e^{2\pi\I v_\beta/w} ) \,
\left( \frac{2\pi\I t}{w} \right)^{2n-2}
+\frac{1}{12}\( \frac{\pi\I v_\beta}{w} +\log(2\pi) \)
\]
\nn\\
&&
+\frac{\chi_X}{2}
\sum_{n\geq 2} \frac{B_{2n} B_{2n-2}  }{2n(2n-2) (2n-2)! }
\left( \frac{2\pi\I t}{w} \right)^{2n-2} -\frac{\chi_X}{24}\, \log \frac{w}{t}\, ,
\label{logtauexp}
\eea
where the last term in the first line originates from \eqref{In0}. Dropping irrelevant constant terms
and identifying the topological string coupling and K\"ahler moduli as\footnote{Note that
the condition $\beta\in B^+$ is consistent with the K\"ahler cone condition $\Im z_\beta>0$.}
\be
\lambda= \frac{2\pi\I t}{w},
\qquad
z_\beta=\frac{v_\beta}{w},
\label{ident-zlam}
\ee
we see that \eqref{logtauexp}
coincides with \eqref{Ftopexp} with genus zero (proportional to $1/\lambda^2$) terms subtracted.
The identification of the genus-1 term, however, requires
\be
\label{linpred}
c_i = -  \sum_{\beta\in B^+} n_\beta \beta_i.
\ee
In principle, one could check this prediction by computing the coefficients $c_i=\int_{\gamma_i} c_2(TX)$,
however the integrals are ill-defined since the  divisors $\gamma_i$ are non-compact,
and it is unclear how to regularize it. Nonetheless, assuming that there exists a choice of
regularization such that \eqref{linpred} holds true, we conclude that
the logarithm of the $\tau$ function for vanishing fiber coordinates provides a possible non-perturbative
completion of the holomorphic topological string partition function $F_{\rm top}(z,\lambda)$.
Clearly, it would be interesting to
have a deeper understanding of the coincidence
between these two asymptotic series.

Note that the $\tau$ function considered here corresponds to the principle value prescription of \S\ref{sec_pass2}.
If we expanded instead the $\tau$ function corresponding to a different solution, we would get a few additional terms.
For example, parametrizing the ambiguity by  $(\eps_+,\eps_-,\eps_0)$ as in \S\ref{subsec-compare},
one finds that these terms are given by
\be
-\eps_0\,\frac{ \chi_X\zeta(3)}{4\pi^2}\, \frac{t}{w}
+\sum_{\beta\in B^+} n_\beta\( (\eps_+ +\eps_-)\,\frac{\zeta(3)}{4\pi^2}\, \frac{t}{w} -(\eps_+ -\eps_-)\,\frac{\pi\I\vbt}{24 w}\).
\label{deltau}
\ee
The first two terms scale as $\lambda$ and are absent in the expansion of the topological string partition function.
This fact can be used as an additional argument for setting $\eps_0=\eps_++\eps_-=0$.
If \eqref{linpred} indeed holds,
then the last term in \eqref{deltau} should also vanish.
This would imply that all $\eps$'s vanish
leaving us again with the solution corresponding to the principle value prescription.

\subsection{Plebanski potential and holomorphic prepotential}

In \cite{Bridgeland:2019fbi} it was shown that for uncoupled BPS structures, there exists a locally
defined function $\cF(z)$, called prepotential,
related to the Plebanski potential $W$ by  the following equations,
\be
\frac{\p^3 \cF}{\p z^a\p z^b\p z^c}=\left.\frac{\p^3 W}{\p \theta^a\p\theta^b\p\theta^c}\right|_{\theta=0},
\qquad \forall a,b,c=1,\dots, 2N.
\label{defprep}
\ee
For the conifold case, the prepotential was calculated in \cite{Bridgeland:2019fbi} using the Plebanski potential
given by \eqref{WBr}, and it was observed that under the identification \eqref{ident-zlam},
it captures the world-sheet instanton part of the genus zero topological string amplitude $F_0(z)$
(which also coincides with the holomorphic prepotential in $\cN=2$ supergravity).
Since the Plebanski potential \eqref{resW-pv} for our preferred solution differs from \eqref{WBr},
it is natural to ask how the prepotential
is affected. A simple calculation gives for our solution\footnote{Equation \eqref{defprep} defines the prepotential
only up to quadratic terms in the moduli; we added the second term by hand so as to reach
a perfect match with $F_0$.}
\be
\cF=(2\pi\I)^2\sum_{\beta\in B^+}n_\beta\, \frac{\vbt^3}{6w}-\frac{\chi_X\zeta(3)}{4\pi\I}\, w^2
+\frac{w^2}{2\pi\I}\sum_{\beta\in B^+}n_\beta\,\Li_3\(e^{2\pi\I \vbt/w}\).
\ee
Using the identifications \eqref{ident-zlam} and comparing with the genus zero part of $F_{\rm hol}$ \eqref{Ftopexp},
one finds that
\be
\cF=-2\pi\I(t/\lambda)^2 \,F_0,
\ee
provided the intersection numbers are given by
\be
\label{kappapred}
\kappa_{ijk}=\sum_{\beta\in B^+} n_\beta \, \beta_i\beta_j\beta_k.
\ee
Like the coefficients $c_i$, the intersection numbers $\kappa_{ijk}=\gamma_i\cap\gamma_j\cap\gamma_k$ are {\it a priori}
ill-defined since the divisors $\gamma_i$ are non-compact.  Assuming that there is a choice of
regularization such that \eqref{kappapred} holds true, it therefore appears that the {\it full}
all-genus topological string partition function can also be extracted from the functions $\tau$ and $W$,
which are in turn determined by the BPS invariants. Unfortunately, the underlying reason for these coincidences
remains mysterious to us.

\appendix

\section{Special functions and useful identities}
\label{ap-functions}

In this appendix we collect various useful definitions, properties and identities (some of which were
already stated in \cite[\S B]{Alexandrov:2021wxu}).

\subsubsection*{Bernoulli polynomials}

The Bernoulli polynomials have the following generating function
\be
\sum_{n=0}^\infty  \frac{x^n}{n!}\, B_n(\xi)= \frac{x\, e^{\xi x}}{e^x-1}.
\label{genBn}
\ee
They have the following symmetry property
\be
B_n(1-x)=(-1)^n B_n(x)
\label{symBn}
\ee
and at $x=0$ they reduce to Bernoulli numbers $B_n$.
Importantly, the Bernoulli polynomials arise in the inversion formula for polylogarithms: namely,
for any $n\geq 2$ and $x\in \IC$, or $n\geq 0$ and $x\in\IC\backslash\IZ$,
\be
\Li_n(e^{2\pi\I x})+(-1)^n\Li_n(e^{-2\pi\I x})=-\frac{(2\pi\I)^n}{n!}\, B_n([x]),
\label{Li-ident}
\ee
where we use the principal branch definition of $\Li_s(z)$, and define
\be
[x]=
\begin{cases}
x-\fl{\Re x},  & \quad \mbox{if } \Im x\ge 0,
\\
x+\fl{-\Re x}+1,  & \quad \mbox{if } \Im x< 0.
\end{cases}
\label{defbr}
\ee
Note that for $\Im x\ne 0$ or $x\notin \IZ$, the bracket satisfies $[-x]=1-[x]$, consistently
with \eqref{Li-ident} and \eqref{symBn}. For integer $n<0$, we have instead
\be
\Li_n(e^{2\pi\I x})+(-1)^n\Li_n(e^{-2\pi\I x})= 0
\label{Li-ident2}
\ee
for all $x\in\IC\backslash\IZ$.

\subsubsection*{Useful identities}

For $d\in\IZ$, $0<\Re z<\Re \omega_1$, $\Im(z/\omega_1)>0$ and the contour $C$ going along the real axis but avoiding the origin from above,
one has \cite[Eqs.(36),(38)]{Bridgeland:2017vbr}
\begin{subequations}
\bea
\int_C \de s\,\frac{e^{zs}\, s^{-d}}{e^{\omega_1s}-1}&=&\(\frac{\omega_1}{2\pi\I}\)^{d-1}\Li_d\(e^{2\pi\I z/\omega_1}\),
\label{evalI-C1}
\\
\int_C \de s\,\frac{e^{(z+\omega_1)s}\, s^{1-d}}{(e^{\omega_1s}-1)^2}&=&
-\frac{\de}{\de\omega_1}\[\(\frac{\omega_1}{2\pi\I}\)^{d-1}\Li_d\(e^{2\pi\I z/\omega_1}\)\],
\label{evalI-C2}\\
\int_C \de s\,\frac{\(e^{\omega_1s}+1\)e^{(z+\omega_1)s}\, s^{2-d}}{(e^{\omega_1s}-1)^3}&=&
\frac{\de^2}{\de\omega_1^2}\[\(\frac{\omega_1}{2\pi\I}\)^{d-1}\Li_d\(e^{2\pi\I z/\omega_1}\)\].
\label{evalI-C3}
\eea
\label{evalI-C}
\end{subequations}
For $\alpha\in\IC\backslash\IZ$, one has \cite[1.445.6]{GradRyzh}
\bea
\sum_{k=1}^\infty \frac{\cos(2\pi kx)}{k^2-\alpha^2}&=&\frac{1}{2\alpha^2}-
\frac{\pi}{2\alpha}\, \frac{\cos(2\pi \alpha(\hf-[x]))}{\sin(\pi\alpha)}\, .
\label{sumcos0}
\eea
\label{sumsincos}

\subsubsection*{Generalized Gamma and Barnes functions}

The generalized Gamma function $\Lambda(z,\eta)$ and generalized Barnes functions $\Upsilon(z,\eta)$
are both functions on $\IC^\times\times\IC$ defined by \cite{Barbieri:2018swu,Alexandrov:2021wxu}
\bea
\Lambda(z,\eta)& =&\frac{e^z\,\Gamma(z+\eta)}{\sqrt{2\pi} z^{z+\eta-1/2}}\, .
\label{defLambda}
\\
\Upsilon(z,\eta) &=& \frac{e^{\frac34z^2- \zeta'(-1)}\, G(z+\eta+1)  }
{(2\pi)^{z/2} \,z^{\frac12 z^2} \bigl[\Gamma(z+\eta) \bigr]^\eta}\, ,
\label{defUpsilon}
\eea
where $G(z)$ is the Barnes function (see e.g. \cite{Vigneras1979}). The two functions are related to each other through
\be
\label{dUps}
\frac{\p}{\p z}  \log \Upsilon(z,\eta) = z \frac{\p}{\p z}  \log \Lambda(z,\eta).
\ee
They have the following integral representations analogous to the first Binet formula
for Gamma function and valid for $\Re z>0$, $\Re (z+\eta)>0$ \cite[\S B]{Alexandrov:2021wxu}
\be
\log \Lambda(z,\eta)=
\int_0^\infty \frac{\de s}{s} \(\eta-\hf-\frac{1}{s}+\frac{e^{(1-\eta)s}}{e^s-1} \) e^{-z s}\, ,
\label{Lam-Binet}
\ee
\be
\log\Upsilon(z,\eta) =
\int_0^\infty \frac{\de s}{s} \( \frac{1}{s^2} - \frac12 \,B_2(\eta)
- \frac{\eta(e^s-1)+1}{(e^s-1)^2}\,e^{(1-\eta)s} \)  e^{-z s} -\frac{1}{2} \,B_2(\eta) \log z.
\label{BinetUps1}
\ee

\section{Special functions relevant for Bridgeland's solution}
\label{ap-Br}

In this appendix we define a class of functions which appear in the solution discussed in \S\ref{sec_pass1}.
We also provide their relations to multiple sine functions, quantum dilogarithm
and present new integral identities for some of these functions.

We define the following set of functions
\be
\cF_{n,m}(z|\omega_1,\omega_2)=(-1)^{n+m}\int_C \frac{\de s}{s}\, e^{z s}\, \Li_{-n}(e^{-\omega_1 s})\, \Li_{-m}(e^{-\omega_2 s}),
\label{deffunF}
\ee
where the contour $C$ follows the real axis from $-\infty$ to $\infty$ avoiding the origin from above.
We consider these functions for $n,m$ non-negative integers, in which case the integral is convergent provided
$\Re(\omega_1),\Re(\omega_2)>0$ and $0<\Re(z)<\Re(\omega_1+\omega_2)$.
For reference, we list
\be
\Li_0(x)=\frac{x}{1-x}\, ,
\qquad
\Li_{-1}(x)=\frac{x}{(1-x)^2}\, ,
\qquad
\Li_{-2}(x)=\frac{x(1+x)}{(1-x)^3}\, .
\ee
In the special cases $(n,m)=(0,0)$ and $(1,0)$, one reproduces the functions $\log F(z|\omega_1,\omega_2)$
and $\log G(z|\omega_1,\omega_2)$, respectively, introduced in \cite[\S 4]{Bridgeland:2017vbr}
in terms of the multiple sine functions $S_r(z|\omega_1,\dots,\omega_r)$,
\bea
e^{\cF_{0,0}(z|\omega_1,\omega_2)} &=& F(z|\omega_1,\omega_2):= e^{-\frac{\pi\I}{2}\, B_{2,2}(z|\omega_1,\omega_2)}S_2(z|\omega_1,\omega_2),
\label{defFfun}
\\
e^{\cF_{1,0}(z|\omega_1,\omega_2)}&=& G(z|\omega_1,\omega_2):= e^{\frac{\pi\I}{6}\, B_{3,3}(z+\omega_1|\omega_1,\omega_1,\omega_2)}S_3(z+\omega_1|\omega_1,\omega_1,\omega_2).
\label{defGfun}
\eea
Here, the  multiple Bernoulli   polynomials $B_{r,m}$ are defined through the generating function generalizing \eqref{genBn},
\be
\sum_{n=0}^\infty  \frac{x^n}{n!}\, B_{r,n}(\xi|\omega_1,\dots,\omega_r)=
\frac{x^r\, e^{\xi x}}{\prod_{j=1}^r (e^{\omega_j x}-1)}\, .
\label{genBn-r}
\ee
In particular, one has
\bea
B_{2,2}(z|\omega_1,\omega_2)&=&\frac{z^2}{\omega_1\omega_2}-\(\frac{1}{\omega_1}+\frac{1}{\omega_2}\)z
+\frac16\(\frac{\omega_2}{\omega_1}+\frac{\omega_1}{\omega_2}\)+\hf\, ,
\label{B22}
\\
B_{3,3}(z|\omega_1,\omega_2,\omega_3)&=&\frac{z^3}{\omega_1\omega_2\omega_3}
-\frac{3z^2(\omega_1+\omega_2+\omega_3)}{2\omega_1\omega_2\omega_3}
+\frac{\omega_1^2+\omega_2^2+\omega_3^2+3(\omega_1\omega_2+\omega_1\omega_3+\omega_2\omega_3)}{2\omega_1\omega_2\omega_3} z\nn\\
&&- \frac{(\omega_1+\omega_2+\omega_3)(\omega_1\omega_2+\omega_1\omega_3+\omega_2\omega_3)}{4 \omega_1\omega_2\omega_3}\, .
\nn
\eea
We refer to \cite{Kurokawa2003MultipleSF} for the definition and main properties of
the multiple sine functions, and to \cite{Narukawa2003TheMP} for an  integral representation
established for these functions, which implies the relations \eqref{defFfun}, \eqref{defGfun}.
Note that all these functions are invariant under a simultaneous rescaling of the arguments
$z$ and $\omega_i$.

\subsubsection*{Double sine, Faddeev's quantum dilogarithm and Woronowicz integral}

It  will be useful to express the function $F(z|\omega_1,\omega_2)$ in terms of Faddeev's quantum dilogarithm \cite{Faddeev:1995nb},
which will allow us to derive a new integral representation for  this function.
The quantum dilogarithm is defined by
\be
\Phi_b(x) = \exp\left[ \frac14 \int_{C} \frac{e^{-2\I x t}}{\sinh(b\, t) \sinh(t/b)} \frac{\de t}{t} \right].
\ee
Comparing this integral with \eqref{deffunF} for $m=n=0$ where the integration variable $s$
is changed to $t=\frac{s}{2}\, \sqrt{\omega_1\omega_2}$, we get
\be
F(z|\omega_1,\omega_2) = \Phi_{\sqrt{\omega_1/\omega_2}}\left( \frac{\I}{\sqrt{\omega_1\omega_2}}
\(z-\frac{\omega_1+\omega_2}{2}\) \right) .
\ee

On the other hand, $\Phi_b$ admits a different integral representation \cite{Garoufalidis:2020pax}
originally due to Woronowicz \cite{woronowicz2000quantum},
\be
\label{WoroPhi}
\Phi_b(x) = \exp\left( \frac{\I}{2\pi}\, W_{1/b^2}(2\pi b x)\right),
\qquad
W_\theta(z) = \int_{\IR} \frac{ \log(1+e^{\theta\xi})}{1+e^{\xi-z}} \de \xi \, .
\ee
Upon folding the integral over $\IR$ onto the positive axis, one gets
\be
W_{\theta}(z) =-\theta\, \Li_2(-e^{z}) +
\int_0^{\infty} \log(1+e^{-\theta\xi})
\left( \frac{1}{1+e^{\xi-z}} + \frac{1}{1+e^{-\xi-z}} \right) \de\xi\, .
\ee
Setting $b^2=\omega_1/\omega_2$ and $\xi=2\pi s \,\omega_1/\omega_2$, one arrives at
the following integral representation,
\bea
\log F(z|\omega_1,\omega_2) &=&
\frac{\omega_2}{2\pi\I \omega_1} \Li_2\left( e^{\frac{2\pi\I}{\omega_2}(z-\frac{\omega_1}{2})}\right)
\\&&+
\frac{\I\omega_1}{\omega_2} \int_0^{\infty}
\left( \frac{ \log(1+e^{-2\pi s})}{1-e^{\frac{2\pi }{\omega_2}
(\omega_1 s- \I (z-\frac{\omega_1}{2}) )}}
+ \frac{ \log(1+e^{-2\pi s})}{1-e^{-\frac{2\pi}{\omega_2}
( \omega_1 s+ \I (z-\frac{\omega_1}{2}))}}
 \right) \de s \, .
\nn
\eea
In order to remove the half-period shift of $z$, we shift the integration variable $s\mapsto s\pm \I/2$ in
each of the two terms on the second line, which gives
\bea
\log F(z|\omega_1,\omega_2) &=&
\frac{\omega_2}{2\pi\I \omega_1} \Li_2\left( e^{\frac{2\pi\I}{\omega_2}(z-\frac{\omega_1}{2})}\right)
+
\sum_{\eps=\pm}\frac{\I\omega_1}{\omega_2} \int_{\eps \I /2}^{\infty}
\frac{\log(1-e^{-2\pi s})}{1-e^{\frac{2\pi }{\omega_2}(\eps \omega_1 s - \I z) }}
\,\de s.
\eea
Further decomposing $ \int_{\eps \I /2}^\infty =  \int_{0}^\infty  - \int_0^{\eps \I /2}$ and
sending $s\mapsto -s$ in the integral from $0$ to $-\I/2$, this can be rewritten as
\bea
\log F(z|\omega_1,\omega_2) &=&\frac{\omega_2}{2\pi\I \omega_1} \Li_2\left( e^{\frac{2\pi\I}{\omega_2}(z-\frac{\omega_1}{2})}\right)
+\frac{2\pi \I\omega_1}{\omega_2} \int_0^{\I /2}
\frac{\(s-\frac{\I}{2}\)\de s}{1-e^{\frac{2\pi }{\omega_2}(\omega_1 s- \I z) }}
\\
&&+\frac{\I\omega_1}{\omega_2}\int_0^\infty\de s\, \log\(1-e^{-2\pi s}\)\[\frac{1}
{1-e^{-\frac{2\pi}{\omega_2}(\omega_1 s+\I z)}}
+\frac{1}{1-e^{\frac{2\pi}{\omega_2}(\omega_1 s-\I z)}}\] .
\nn
\eea
Finally, using the identity (easily established by integrating by parts)
\be
\int_0^x \frac{(s-x)\de s}{1-e^{as-y}}=-\frac{x}{a}\, \log\(1-e^y\)+\frac{1}{a^2}\, \(\Li_2\(e^{y-ax}\)-\Li_2\(e^y\)\),
\ee
one obtains
\bea
\label{WoroF}
\log F(z|\omega_1,\omega_2) &=& \hf\, \log\(1-e^{2\pi\I z/\omega_2}\)
+\frac{\omega_2}{2\pi\I \omega_1}\, \Li_2\(e^{2\pi\I z/\omega_2}\)
\\
&&+\frac{\I\omega_1}{\omega_2}\int_0^\infty\de s\, \log\(1-e^{-2\pi s}\)\[\frac{1}
{1-e^{-\frac{2\pi}{\omega_2}(\omega_1 s+\I z)}}
+\frac{1}{1-e^{\frac{2\pi}{\omega_2}(\omega_1 s-\I z)}}\].
\nn
\eea
This integral representation plays a key role in relating two types of solution of the RH problem as discussed in \S\ref{subsec-compare}
because the last line can be equivalently rewritten in terms of the functions \eqref{deffunJ} as
\be
\frac{\I\omega_2}{2\omega_1}\(\cJ_{1,0}\(-\frac{\omega_1}{\omega_2},-\frac{z}{\omega_2},0,0\)-
\cJ_{1,0}\(-\frac{\omega_1}{\omega_2},\frac{z}{\omega_2},0,0\)\)-\frac{\pi\I\omega_2}{12\omega_1}\, .
\ee

\subsubsection*{Woronowicz-type integral representations for the triple sine function}

By comparing Bridgeland's expressions for $\cY_0$ and $\log\tau$ with our solution,
we arrive at the following conjecture, supported
by numerical checks for random values of the arguments:

\begin{conj}
\label{conj-sine}
One has the following integral representations for the triple sine function
for two coinciding periods,
\bea
\log S_3  (z|\omega_1,\omega_1,\omega_2) &=&
 -\frac{\pi\I}{6}\, B_{3,3}(z|\omega_1,\omega_1,\omega_2)
+\hf\(1-\frac{z}{\omega_1}\) \log\(1-e^{2\pi\I z/\omega_1}\)
\nn\\
&&
-\frac{1}{2\pi\I}\(\frac{z-\omega_1}{ \omega_2}+\hf\)\, \Li_2\(e^{2\pi\I z/\omega_1}\)
-\frac{\omega_1}{2\pi^2\omega_2}\, \Li_3\(e^{2\pi\I z/\omega_1}\)
\label{WoroG1}\\
&&
+\frac{\omega_2^2}{\omega_1^2}\int_0^\infty\de s\, \log\(1-e^{-2\pi s}\)
\[\frac{s+\frac{\I(z-\omega_1)}{\omega_2}}{e^{-\frac{2\pi}{\omega_1}(\omega_2 s+\I z)}-1}
-\frac{s-\frac{\I(z-\omega_1)}{\omega_2}}{e^{\frac{2\pi}{\omega_1}(\omega_2 s-\I z)}-1}\]
\nn
\eea
and
\bea
\log S_3  (z|\omega_1,\omega_1,\omega_2) &=&
 -\frac{\pi\I}{6}\, B_{3,3}(z|\omega_1,\omega_1,\omega_2)
+\frac{\omega_2^2}{4\pi^2\omega_1^2}\, \Li_3\(e^{2\pi\I z/\omega_2}\)-\frac{1}{12}\, \log\(1-e^{2\pi\I(z-\omega_1)/\omega_2 }\)
\nn\\
&-&
\frac{\omega_1}{4\pi \omega_2}\int_0^\infty \de s\,
\frac{\bigl(\Li_2 \(e^{-2\pi s}\) -2\pi s\log\(1-e^{-2\pi s}\)\bigr)\, \sinh(2\pi s\, \omega_1/\omega_2) }
{\sinh\(\frac{\pi}{\omega_2}(\omega_1 (s-\I)+\I z)\)\sinh\(\frac{\pi}{\omega_2}(\omega_1 (s+\I)-\I z)\)}\, ,
\label{WoroG2}
\eea
where the triple Bernoulli polynomial $B_{3,3}$ is given in \eqref{B22}.
\end{conj}
\noindent
Note that the integrals in \eqref{WoroG1} and \eqref{WoroG2} can be expressed as linear combinations
of functions $\cJ_{m,n}$ and $\cI_n$ introduced in Appendix \ref{ap-new}.


\section{Special functions relevant for the new solution}
\label{ap-new}

In this appendix we define a class of functions which appear in the solution obtained in \S\ref{sec_pass2}.
We also describe their analytic properties, establish various identities between them and derive their asymptotic expansion.

\subsubsection*{Definition and analytic properties}

For $m$ positive integer and $n$ non-negative integer, we define
\be
\cJ_{m,n}(x,y,\xi,\eta)=(-1)^{n+m+1}\int_0^\infty \de s\, \frac{s^n \,
\Li_m\(e^{-2\pi(1- \xi/x)s+2\pi\I (\eta-\xi y/x) }\)}{\tanh(\pi(s-\I y)/x)}\, .
\label{deffunJ}
\ee
The integral converges provided $\Re(\xi/x)<1$ and $kx+\I y \notin\IR^+$ for any $k\in\IZ$.
We also introduce a slightly modified version of these functions which removes the pole at $y=0$ appearing for $\cJ_{m,0}$.
Namely, we define
\be
\tcJ_{m,n}(x,\xi)=(-1)^{n+m+1}\int_0^\infty \de s\,s^n\(
\frac{\Li_m\(e^{-2\pi(1- \xi/x)s}\)}{\tanh(\pi s/x)}
-\frac{x}{\pi s}\,\Li_m\(e^{-2\pi s}\)\) ,
\label{deffunJ0}
\ee
which converges for $\Re(\xi/x)<1$ and $\Re(x)\neq 0$.

The analytic continuation to the full complex $\xi$ plane
is easily obtained by rewriting the integrals as
\bea
\cJ_{m,n}(x,y,\xi,\eta)&=&\veps(-1)^{n+m} \[
2 \int_0^\infty \de s\, \frac{s^n \,\Li_m\(e^{-2\pi(1- \xi/x)s+2\pi\I( \eta-\xi y/x) }\)}{1-e^{2\pi\veps (s-\I y)/x}}
\right.
\nn\\
&&\left.\qquad
-\frac{n!}{(2\pi)^{n+1}}\, \frac{\Li_{m+n+1}\(e^{2\pi\I ( \eta-\xi y/x)}\)}{(1- \xi/x)^{n+1}}\] ,
\label{deffunJ-an}
\eea
\bea
\tcJ_{m,n}(x,\xi)&=&\veps(-1)^{n+m}\[2\int_0^\infty \de s\,s^n\(
\frac{\Li_m\(e^{-2\pi(1- \xi/x)s}\)}{1-e^{2\pi\veps s/x}}
-\frac{x}{\pi s}\,\Li_m\(e^{-2\pi s}\)\)
\right.
\nn\\
&&\left.\qquad
-\frac{n!}{(2\pi)^{n+1}}\, \frac{\zeta(m+n+1)}{(1- \xi/x)^{n+1}}\] ,
\label{deffuntJ-an}
\eea
where  we have set $\veps=\sgn(\Re x)$.
The integrals are now convergent for any $\xi$, whereas the pole at $\xi=x$ reflects the original obstruction.

The analytic structure in the variables $x$ and $y$ is more complicated.  As indicated above,  the integrand in \eqref{deffunJ}
has a series of poles at $s=\I (y+kx)$, $k\in\IZ$,
which can cross the integration contour as  $x$ and $y$ are varied
As a result, $\cJ_{m,n}$ is a multi-valued function: it jumps when the arguments $(x,y)$ cross one of the hypersurfaces in $\IC^2$
given by
\be
\Re (y+k x)=0, \qquad \Im(y+k x) <0.
\label{poles}
\ee
The jump is given by the following Proposition which follows from the residue evaluation
\begin{proposition}
\label{porp-jump}
Let $x_k,y_k$ satisfy \eqref{poles}. Then the difference between values of $\cJ_{m,n}$
evaluated from the two sides of the hypersurface with $\pm \Re(y+kx)>0$, respectively, is
\be
\Delta\cJ_{m,n}(x_k,y_k,\xi,\eta)=
2\I^{n-1} (-1)^{n+m} x_k(y_k+kx_k)^n \,\Li_m\(e^{2\pi\I(\eta -y_k++k(\xi-x_k))}\).
\nn
\ee
\end{proposition}

A similar picture holds for $\tcJ_{m,n}$ with the difference that now all poles hit the integration contour
at the same hypersurface. It is described by
\begin{proposition}
\label{porp-jump0}
Let $x\in \I \IR^-$. Then the difference between values of $\tcJ_{m,n}$
evaluated from the two sides of this line with $\pm \Re x>0$, respectively, is
\be
\Delta\tcJ_{m,n}(x,\xi)=
2\I^{n-1} (-1)^{m+n}  x^{n+1}\sum_{k=1}^\infty k^n\Li_m\(e^{2\pi k(\xi-x)}\).
\nn
\ee
\end{proposition}

\subsubsection*{Identities}

By differentiating under the sign of the integral and integrating by parts,
it is straightforward to prove
\begin{proposition}
\label{prop-ident}
\bea
\p_\eta\cJ_{m,n}&=&-2\pi\I \cJ_{m-1,n},
\nn\\
\p_\xi\cJ_{m,n}&=&\frac{2\pi}{x}\, \cJ_{m-1,n+1}+\frac{2\pi\I y}{x}\, \cJ_{m-1,n},
\nn\\
\p_y\cJ_{m,n}&=&-\I n \cJ_{m,n-1}+2\pi\I \cJ_{m-1,n}
+\delta_{n,0}(-1)^{m+n}\,\frac{\Li_m\(e^{2\pi\I (\eta-\xi y/x) }\)}{\tan(\pi y/x)}\, ,
\nn\\
\p_x\cJ_{m,n}&=&
-\frac{2\pi}{x}\, \cJ_{m-1,n+1}-\frac{2\pi\I y}{x}\, \cJ_{m-1,n}
+\frac{n+1}{x}\, \cJ_{m,n}+\frac{\I n y}{x}\, \cJ_{m,n-1}
\nn\\
&&
+\delta_{n,0}(-1)^{m+n+1}\,\frac{y}{x}\,\frac{\Li_m\(e^{2\pi\I (\eta-\xi y/x) }\)}{\tan(\pi y/x)}\, ,
\nn
\eea
and
\bea
\p_x\tcJ_{m,n}&=&\frac{n+1}{x}\,\tcJ_{m,n}-\frac{2\pi}{x}\, \tcJ_{m-1,n+1},
\nn\\
(\p_x+\p_\xi)\tcJ_{m,n}&=&\frac{n+1}{x}\, \tcJ_{m,n}
+(-1)^{m+n+1}\,\frac{n!\,\zeta(m+n)}{\pi (2\pi)^n}\,.
\nn
\eea
\end{proposition}
\noindent
In particular, we have the following relations
\bea
(x\p_x+y\p_y)\cJ_{m,n}&=&(n+1)\cJ_{m,n}-2\pi \cJ_{m-1,n+1},
\nn
\\
(\p_y+\p_\eta)\cJ_{m,n}&=&-\I n \cJ_{m,n-1}
+\delta_{n,0}(-1)^{m+n}\,\frac{\Li_m\(e^{2\pi\I (\eta-\xi y/x) }\)}{\tan(\pi y/x)}\, ,
\label{ident-more}\\
(\p_x+\p_\xi)\cJ_{m,n}&=&\frac{n+1}{x}\, \cJ_{m,n}+\frac{\I n y}{x}\, \cJ_{m,n-1}
+\delta_{n,0}(-1)^{m+n+1}\,\frac{y}{x}\,\frac{\Li_m\(e^{2\pi\I (\eta-\xi y/x) }\)}{\tan(\pi y/x)}\, .
\nn
\eea

\subsubsection*{Asymptotic expansion}

In the discussion of the $\tau$ function we use the asymptotic expansions of the functions $\cJ_{m,n}$
and $\tcJ_{m,n}$ at vanishing values of $\xi$ and $\eta$. They are given by the following
\begin{proposition}
\label{porp-expand}
At large $x$ and $y$ such that $x/y$ is kept fixed, one has the following asymptotic expansions
\bea
\cJ_{m,n}(x,y,0,0) +\cJ_{m,n}(x,-y,0,0)&=&
\frac{4(-1)^{n+m}}{(2\pi)^{n+1}}
\sum_{k=1}^\infty \frac{(n+2k-1)!}{ (2k-1)!}\, \frac{\zeta(m+n+2k)}{x^{2k-1}}\,\Li_{1-2k}\(e^{2\pi\I y/x}\),
\nn\\
\tcJ_{m,n}(x,0) &=& 2(-1)^{n+m+1}\sum_{k=1}^\infty \frac{(n+2k-1)!}{(2\pi)^{n+1}(2k)!}\, \zeta(m+n+2k)\,B_{2k}\, x^{1-2k}.
\nn
\eea
\end{proposition}
\begin{proof}
Note that in the limit we are interested in, we have
\be
\frac{1}{\tanh(\pi(s-\I y)/x)}=1+2\Li_0\( e^{-2\pi(s-\I y)/x}\)
=1+2\sum_{k=0}^\infty \frac{(-1)^k}{k!}\, \Li_{-k}\(e^{2\pi\I y/x}\)\(\frac{2\pi s}{x}\)^k.
\label{expcoth}
\ee
Therefore, after setting $\xi=\eta=0$ in \eqref{deffunJ} and using
\be
\int_0^\infty \de s\, s^n \,\Li_m\(e^{-2\pi s}\)=\frac{n!}{(2\pi)^{n+1}}\, \zeta(m+n+1),
\label{intLi}
\ee
we obtain the following asymptotic expansion
\be
\begin{split}
\cJ_{m,n}(x,y,0,0)=&\,
(-1)^{n+m+1}\[\frac{n!}{(2\pi)^{n+1}}\, \zeta(m+n+1)
\right.
\\
&\, \left.
+2\sum_{k=0}^\infty \frac{(-1)^k(n+k)!}{(2\pi)^{n+1} k!}\, \zeta(m+n+k+1)\,\Li_{-k}\(e^{2\pi\I y/x}\)x^{-k}\].
\end{split}
\label{exp-cJone}
\ee
To combine two such expansions evaluated at $\pm y$, we use \eqref{Li-ident} and \eqref{Li-ident2}.
As a result, the first term in \eqref{exp-cJone} is cancelled by the $k=0$ contribution, the terms with $k$ even vanish,
and terms with $k$ odd are doubled.
This gives exactly the first statement of the proposition.

To get the second statement, we follow the same steps except that, instead of \eqref{expcoth}, one starts with the expansion
\be
\frac{1}{\tanh(\pi s/x)}
=\frac{x}{\pi s}+2\sum_{k=1}^\infty \frac{B_{2k}}{(2k)!}\,\(\frac{2\pi s}{x}\)^{2k-1}.
\label{expcoth0}
\ee
Its first term cancels the last term in \eqref{deffunJ0} and the remaining series can be easily integrated using \eqref{intLi}
which gives the desired statement.
\end{proof}

\medskip

Finally, we also introduce functions appearing in the expression for the $\tau$ function
\be
\cI_n(x,\xi,\eta)=\int_0^{\pi\I x}\de y\(\frac{\Li_n\(e^{-2\xi y+2\pi\I \eta}\)}{\tanh(y)}-\frac{\Li_n\(e^{2\pi\I \eta}\)}{y}\)
+\Li_n\(e^{2\pi\I \eta}\)\log x.
\label{def-cIn}
\ee
Note their special value at $\xi=0$:
\be
\label{In0}
\begin{split}
\cI_n(x,0,\eta)=&\,\Bigl(\log\(1-e^{2\pi\I x}\) -\pi\I x-\log(-2\pi\I)\Bigr)\Li_n\(e^{2\pi\I \eta}\) .
\end{split}
\ee

\section{Details on summation of finite case results}
\label{ap-conifold}

In this appendix we analyze the approach towards definition of the Darboux coordinates \eqref{defcY} and their
$\tau$ function based on the attempt to perform the sum over charges in the results emerging in the finite case,
\eqref{Xuncoupled} and \eqref{taufun-res2}.
Throughout we use the shorthand notations $\hw=w/t$ and $\hvb=\vbt/t$,
and restrict to the region in the moduli space where $\Re \hw<\Re \hvb<0$ for all $\beta\in B^+$
which allows us to apply the integral representations \eqref{Lam-Binet} and \eqref{BinetUps1}.
Furthermore, we assume that $\Re\thb\in(0,1)$, while $\Re\vph=0$ so that we can
identify $[\vph]$ with $\vph$ (see \eqref{defbr}).

\subsubsection*{Darboux coordinates}

Using the finite case solution \eqref{Xuncoupled} and the BPS data given in \eqref{OmCY}, one finds that the functions \eqref{defcY}
take the following form
\be
\cY_i=\sum_{\beta\in B^+} n_\beta \beta_i \cYa,
\qquad
\cY_0=\sum_{\beta\in B^+} n_\beta \cYb-\frac{\chi(X)}{2} \cYc,
\label{cY-Br1}
\ee
where we denoted
\begin{subequations}
\bea
\cYa &=&
\sum_{k=1}^{\infty} \log \Lambda(\hvb-k\hw,1-\thb+k\vph)-\sum_{k=0}^{\infty} \log \Lambda(-\hvb-k\hw,\thb+k\vph),
\label{cYa}
\eea
\bea
\cYb &=& -\sum_{k=1}^{\infty} k \log \Lambda(\hvb-k\hw,1-\thb+k\vph)-\sum_{k=0}^{\infty} k\log \Lambda(-\hvb-k\hw,\thb+k\vph),
\label{cYb}\\
\cYb &=& -2\sum_{k=1}^{\infty} k \log \Lambda(-k\hw,k\vph).
\label{cYc}
\eea
\label{cYabc}
\end{subequations}
Under our assumptions,
the real part of the arguments of all $\Lambda$-functions are positive, therefore one can apply the integral representation \eqref{Lam-Binet}.
Substituting it into \eqref{cYabc} and exchanging the sum over $k$ with the integral, one obtains
\begin{subequations}
\bea
\cYa &=&
-\int_0^\infty \frac{\de s}{s}\[
\(\frac{e^{(\thb-\vph)s}}{\(e^{(\hw-\vph)s}-1\)\(e^{s}-1\)}
+\frac{\hf-\frac{1}{s}-\thb+\vph}{e^{\hw s}-1}-\frac{\vph\, e^{\hw s}}{\(e^{\hw s}-1\)^2}\)e^{(\hw-\hvb)s}
\right.
\nn\\
&&\left.
-\(\frac{e^{(1-\thb)s}}{\(e^{(\hw-\vph)s}-1\)\(e^{s}-1\)}
-\frac{\hf+\frac{1}{s}-\thb}{e^{\hw s}-1}-\frac{\vph\, e^{\hw s}}{\(e^{\hw s}-1\)^2}\)e^{\hvb s}\],
\label{cYa2}\\
\cYb &=&
-\int_0^\infty \frac{\de s}{s}\[\(\frac{e^{(\thb-\vph)s}}{\(e^{(\hw-\vph)s}-1\)^2\(e^{s}-1\)}
+\frac{\hf-\frac{1}{s}-\thb}{(e^{\hw s}-1)^2}-\frac{\vph\(e^{\hw s}+1\)}{\(e^{\hw s}-1\)^3}\)e^{(\hw-\hvb)s}
\right.
\nn\\
&&\left.
+\(\frac{e^{(1-\thb-\vph)s}}{\(e^{(\hw-\vph)s}-1\)^2\(e^{s}-1\)}
-\frac{\hf+\frac{1}{s}-\thb}{(e^{\hw s}-1)^2}-\frac{\vph\(e^{\hw s}+1\)}{\(e^{\hw s}-1\)^3}\)e^{(\hw+\hvb)s}\],
\label{cYb2}\\
\cYc &=&
-2\int_0^\infty \frac{\de s}{s}\(\frac{e^{(1-\vph)s}}{(e^{(\hw-\vph)s}-1)^2\(e^{s}-1\)}
-\frac{\hf+\frac{1}{s}}{(e^{\hw s}-1)^2}
-\frac{\vph\(e^{\hw s}+1\)}{(e^{\hw s}-1)^3}\) e^{\hw s}.
\label{cYc2}
\eea
\label{cYabc2}
\end{subequations}

Let us first analyze $\cYa$ for vanishing $\thb$ and $\vph$. It is easy to check that each round bracket
in the integrand behaves near $s=0$ as $1/(12 \hw)$.
While each such term gives rise to a divergent integral, the contributions of the two brackets
cancel against each other so that the whole integral is convergent.
Changing variable $s\mapsto -s$ in the second term, we can recombine the two integrals over $\IR^+$
into a single  integral over the real axis.
Although the integrand has a pole at $s=0$, the cancellation of this divergence in the original integral
implies that the integral should be defined through the principle value prescription. Thus, we have
\be
\begin{split}
\cYa=&\, \Pv\int_{-\infty}^\infty \frac{\de s}{s}\(\frac{1}{e^s-1}+\hf-\frac{1}{s}\)\frac{e^{-\hvb s}}{e^{-\hw s}-1}
\\
=&\, \int_C \frac{\de s}{s}\(\frac{1}{e^s-1}+\hf-\frac{1}{s}\)\frac{e^{-\hvb s}}{e^{-\hw s}-1}-\frac{\pi\I}{12 \hw}\, ,
\end{split}
\label{cYa-03}
\ee
where in the second line the contour $C$ goes along the real axis avoiding the origin from above.
Noticing that the first term in the round brackets gives rise to the function $\cF_{0,0}$ \eqref{deffunF},
while the other two terms can be evaluated using \eqref{evalI-C}, one finally obtains
\be
\cYa= \cF_{0,0}(-\hvb|-\hw, 1)+\frac{\hw}{2\pi\I}\, \Li_2\(e^{2\pi\I \hvb/\hw}\)-\hf \log\( 1-e^{2\pi\I \hvb/\hw}\)-\frac{\pi\I}{12 \hw} .
\label{res-cI1}
\ee
Given the invariance of $\cF_{n,m}$ under simultaneous rescaling of all its arguments, this result coincides
with the function $B(v,w,t)$ in \cite[Thm 5.2]{Bridgeland:2017vbr}.

The situation with $\cYb$ and $\cYc$ is more complicated. In contrast to $\cYa$, the integrands in \eqref{cYb2} and \eqref{cYc2}
behave near $s=0$ as $1/(6\hw^2 s^2)$. Hence both integrals are divergent.
The only exception is the case where the condition \eqref{condchi} is satisfied because it ensures
that in the combination entering $\cY_0$ \eqref{cY-Br1} the divergence is canceled.
Then one can repeat the same steps as in \eqref{cYa3} and arrive at the expression
\bea
\cY_0&=& \sum_{\beta\in B^+} n_\beta
\[\cF_{1,0}(-\hvb|-\hw,1)-\frac{\hw}{2\pi^2}\, \Li_3\(e^{2\pi\I \hvb/\hw}\)
+\frac{1-2\hvb}{4\pi\I}\, \Li_2\(e^{2\pi\I \hvb/\hw}\)
\right.
\\
&&
\left. +\frac{\hv}{2\hw} \log\( 1-e^{2\pi\I \hvb/\hw}\)\]
-\frac{\chi_X}{2}\[\cF_{1,0}(0|-\hw,1)-\frac{\hw\zeta(3)}{2 \pi^2}+\frac{\zeta(2)}{4\pi\I }\]
+\sum_{\beta\in B^+} n_\beta\,\frac{\pi\I \hvb}{12 \hw^2}\, ,
\label{res-cI2} \nn
\eea
where the two square brackets can be seen as coming from $\cYb$ and $\cYc$, respectively, while the last term results from
the residue at $s=0$. For the case of the conifold \eqref{Omconifold},
this result perfectly agrees with the function $D(v,w,t)$ in \cite[Thm 5.2]{Bridgeland:2017vbr}.

For more general CY threefolds which do not obey the condition \eqref{condchi},
to define $\cYb$ and $\cYc$, we have to {\it ignore} the divergence at $s=0$ and {\it postulate}
that they are given by the integrals along the contour $C$ shifted away from the origin. Then $\cYb$ and $\cYc$
equal to the first and second square brackets in \eqref{res-cI2}, respectively.
Note that in this way the last term gets dropped.

After switching on the variables $\thb$ and $\vph$, the same problem with divergence arises already for all three functions \eqref{cYabc2}
and even in the conifold case. For instance, the integrand in $\cYa$ near $s=0$ behaves as
$\frac{\vph(2(\thb \hw-\vph\hvb)-\hw)}{\hw^2(\hw-\vph)s^2}$.
To define all three functions, we proceed as above: we rewrite them as an integral along the real axis and declare that
the integration contour is given by $C$. This gives
\bea
\cYa &=&
\int_C\frac{\de s}{s}\[\frac{e^{-(\hvb-\thb)s}}{\(e^{-(\hw-\vph)s}-1\)\(e^{s}-1\)}
+\frac{\(\hf-\frac{1}{s}-\thb\)e^{-\hw s}}{e^{-\hw s}-1}+\frac{\vph\, e^{-(\hvb+\hw)s}}{\(e^{-\hw s}-1\)^2}\],
\label{cYa3}\\
\cYb &=&
-\int_C\frac{\de s}{s}\[\frac{e^{-(\hvb+\hw-\thb-\vph)s}}{\(e^{-(\hw-\vph)s}-1\)^2\(e^{s}-1\)}
+\frac{\(\hf-\frac{1}{s}-\thb\)e^{-(\hvb+\hw)s}}{\(e^{-\hw s}-1\)^2}+\frac{\vph\, e^{-(\hvb+\hw)s}\(e^{-\hw s}+1\)}{\(e^{-\hw s}-1\)^3}\],
\label{cYb3}
\nn\\
\cYc &=&
-\int_C \frac{\de s}{s}\[\frac{e^{-(\hvb-\vph)s}}{(e^{-(\hw-\vph)s}-1)^2\(e^{s}-1\)}
+\frac{\(\hf-\frac{1}{s}\) e^{-\hw s}}{(e^{-\hw s}-1)^2}
+\frac{\vph\, e^{-\hw s}\(e^{-\hw s}+1\)}{(e^{-\hw s}-1)^3}\]
+\Delta\cJ_2,
\label{cYc3}
\nn
\eea
\label{cYabc3}
where
\be
\Delta\cJ_2=-\int_0^\infty \frac{\de s}{s}\[\frac{e^{-(\hw-\vph)s}}{(e^{-(\hw-\vph)s}-1)^2}-\frac{e^{-\hw s}}{(e^{-\hw s}-1)^2}\].
\ee
The contribution $\Delta\cJ_2$ is still divergent. It arises because the integrand in \eqref{cYc2}
is not antisymmetric with respect to $s\to -s$.
Proceeding as for the divergence at $s=0$, we simply drop this contribution so that one has
$\cYc=\cY^{(2)}_{\beta=0}$.
Finally, computing the integrals using \eqref{evalI-C}, taking into account the homogeneity of the functions $\cF_{n,m}$,
and substituting the resulting expressions into \eqref{cY-Br1}, one arrives at the expressions \eqref{solBridge}
given in the main text.

\subsubsection*{$\tau$ function}

The analysis of the $\tau$ function is similar, but slightly more complicated.
Its expression \eqref{taufun-res2} implies that
\be
\log\tau =  \sum_{\beta\in B^+}n_\beta\, \cT_\beta-\frac{\chi(X)}{2}\, \cT_0,
\label{tau-Bridge1}
\ee
where
\begin{subequations}
\bea
\cT_\beta &=&
\sum_{k=1}^{\infty} \log \Upsilon(\hvb-k\hw,1-\thb+k\vph)+\sum_{k=0}^{\infty} \log \Upsilon(-\hvb-k\hw,\thb+k\vph),
\label{cTa1}
\\
\cT_0 &=& 2\sum_{k=1}^{\infty} \log \Upsilon(-k\hw,k\vph).
\label{cTb1}
\eea
\label{cT1}
\end{subequations}
Using the integral representation \eqref{BinetUps1} and performing the sum over $k$ under the integral,
the two functions  in \eqref{cT1} can be rewritten as
\begin{subequations}
\bea
\cT_\beta &=&
\int_0^\infty\frac{\de s}{s}\[\(\frac{\frac{1}{s^2}-\hf\, B_2(\thb)}{1-e^{\hw s}}
+\frac{\(\thb(1-e^s)-1\)e^{(1-\thb)s}}{(e^s-1)^2(1-e^{(\hw-\vph)s})}
-\frac{\vph\, e^{(1-\thb-\vph+\hw)s}}{(e^s-1)(1-e^{(\hw-\vph)s})^2} \)e^{\hvb s}
\right.
\nn\\
&+&
\(\frac{\frac{1}{s^2}-\hf\, B_2(1-\thb)}{1-e^{\hw s}}
+\frac{\((1-\thb)(1-e^s)-1\)e^{(\thb-\vph)s}}{(e^s-1)^2(1-e^{(\hw-\vph)s})}
-\frac{\vph\, e^{(\thb-\vph)s}}{(e^s-1)(1-e^{(\hw-\vph)s})^2} \)e^{(\hw-\hvb) s}
\nn\\
&+&\left.
\frac{\vph\, (1-2\thb)\,e^{\hw s}}{2(1-e^{\hw s})^2}\(e^{\hvb s}-e^{-\hvb s}\)
-\frac{\vph^2\, e^{\hw s}(e^{\hw s}+1)}{2(1-e^{\hw s})^3}\(e^{\hvb s}+e^{-\hvb s}\)
\]+\cS_\beta,
\label{cTa2}
\\
\cT_0 &=&
2\int_0^\infty\frac{\de s}{s}\[
\(\frac{\frac{1}{s^2}-\frac{1}{12}}{1-e^{\hw s}}
-\frac{e^{(1-\vph)s}}{(e^s-1)^2(1-e^{(\hw-\vph)s})}
-\frac{\vph\, e^{(1-\vph)s}}{(e^s-1)(1-e^{(\hw-\vph)s})^2}\)e^{\hw s}
\right.
\nn\\
&&\left.
+\frac{\vph\, e^{\hw s}}{2(1-e^{\hw s})^2}-\frac{\vph^2\, e^{\hw s}(e^{\hw s}+1)}{2(1-e^{\hw s})^3}
\]+\cS_0,
\label{cTb2}
\eea
\label{cT2}
\end{subequations}
where
\begin{subequations}
\bea
\cS_\beta &=& -\hf\, B_2(\thb)\log(-\hvb)-
\sum_{k=1}^{\infty}\Biggl[ B_2(k\vph)\log((k\hw)^2-\hvb^2)
\Biggr.
\nn\\
&&\left.
+2k\vph\(\log(\hvb-k\hw)-\vth \log\frac{1+\frac{\hvb}{k\hw}}{1-\frac{\hvb}{k\hw}}\)+(\vth^2-\vth)\log((k\hw)^2-\hvb^2)\]
\label{cSa1}
\\
\cS_0 &=& -\sum_{k=1}^{\infty}B_2(k\vph)\log (-k\hw).
\label{cSb1}
\eea
\label{cS1}
\end{subequations}
When  the condition \eqref{condchi} is satisfied and for $\thb=\vph=0$, one has
\be
\begin{split}
\sum_{\beta\in B^+} n_\beta\cS_\beta- \frac{\chi_X}{2}\,\cS_0=&\,
-\frac{1}{12}\sum_{\beta\in B^+} n_\beta\[\log(-\hvb)+\sum_{k=1}^\infty  \log\(1-\frac{\hvb^2}{(k\hw)^2}\)\]
\\
=&\,
-\frac{1}{12}\sum_{\beta\in B^+} n_\beta\[\log\frac{\hw}{2\pi\I }- \frac{\pi\I\vbt}{w}+\log\(1-e^{2\pi\I \vbt/w}\)\].
\end{split}
\label{res-cS-con}
\ee
Furthermore, the integral resulting from the combination $\sum_{\beta\in B^+} n_\beta\cT_\beta-\frac{\chi_X}{2}\,\cT_0$
is convergent and can be treated in the same way as for Darboux coordinates.
Namely, one can rewrite it as an integral over the real axis using the principle
value prescription and evaluate it in terms of the functions $\cF_{n,m}$, using \eqref{evalI-C} and
the identity
\be
\int_0^\infty\frac{\de s}{s}\[\frac{1}{s^2} -\frac{e^{s}}{(e^{s}-1)^2}
+\frac{1}{6\hw s}+\frac{1}{6\(e^{-\hw s}-1\)}\]=\frac{1}{12}\, \log\frac{- w}{2\pi}  -\zeta'(-1).
\ee
The result is
\be
\begin{split}
\log\tau=&\, \sum_{\beta\in B^+} n_\beta\[\cF_{1,0}(-\hvb|-\hw,1)+\frac{w^2}{(2\pi\I t)^2}\,\Li_3\(e^{2\pi\I \vbt/w}\)
+\frac{\pi\I}{12}\, \frac{\vbt}{w}  \]
\\
&\, -\frac{\chi_X}{2}\[\cF_{1,0}(0|-\hw,1)+\frac{w^2}{(2\pi\I t)^2}\,\zeta(3)+\zeta'(-1)\],
\end{split}
\ee
which coincides with \cite[Thm 1.2]{Bridgeland:2017vbr}  for the conifold case, up to an irrelevant constant term.

In more general cases, we again encounter divergences which now arise both in the integrals at $s=0$ and in the series \eqref{cS1}.
Ignoring them, one gets \eqref{taufun-Bridge}.


\providecommand{\href}[2]{#2}\begingroup\raggedright\endgroup

\end{document}